\documentclass[letterpaper, 12 pt]{article}  % % Comment this line out
\usepackage{helvet}
\usepackage[left=2cm, right=2cm, top=2cm, bottom=2.5cm]{geometry}                     

\usepackage{lipsum}

\newcommand\blfootnote[1]{%
  \begingroup
  \renewcommand\thefootnote{}\footnote{#1}%
  \addtocounter{footnote}{-1}%
  \endgroup
}

\usepackage{subcaption}
\usepackage{mwe}
\usepackage{algorithm}
\usepackage{balance} %even columns on last page
\usepackage{algpseudocode}
\usepackage{xcolor,calc}
\usepackage{graphics}
\usepackage{hyperref}
\usepackage{epstopdf}
\usepackage{graphicx}
\usepackage{amsmath} % assumes amsmath package installed
\usepackage{amssymb}  % assumes amsmath package installed
\usepackage{amsthm}
\usepackage[noadjust]{cite}
\usepackage{color}
\usepackage{enumerate}
\usepackage{multirow}
\usepackage{mathtools}
\usepackage{booktabs}
\usepackage{epstopdf}
\usepackage[normalem]{ulem}%%%%%%%%%

\usepackage{url}
\usepackage{bm}
\usepackage{caption}

\usepackage{booktabs}

\newtheorem{theorem}{Theorem}
\newtheorem{remark}{Remark}

\makeatletter
\newsavebox\myboxA
\newsavebox\myboxB
\newlength\mylenA

\newcommand*\xoverline[2][0.75]{%
    \sbox{\myboxA}{$\m@th#2$}%
    \setbox\myboxB\null% Phantom box
    \ht\myboxB=\ht\myboxA%
    \dp\myboxB=\dp\myboxA%
    \wd\myboxB=#1\wd\myboxA% Scale phantom
    \sbox\myboxB{$\m@th\overline{\copy\myboxB}$}%  Overlined phantom
    \setlength\mylenA{\the\wd\myboxA}%   calc width diff
    \addtolength\mylenA{-\the\wd\myboxB}%
    \ifdim\wd\myboxB<\wd\myboxA%
       \rlap{\hskip 0.5\mylenA\usebox\myboxB}{\usebox\myboxA}%
    \else
        \hskip -0.5\mylenA\rlap{\usebox\myboxA}{\hskip 0.5\mylenA\usebox\myboxB}%
    \fi}
\makeatother

\usepackage{tabularx}
\usepackage{booktabs}

\newtheorem{definition}{Definition}
\newtheorem{assumption}{Assumption}

\makeatletter
\let\NAT@parse\undefined
\makeatother

\begin{document}

\title{Robust MPC for Linear Systems with Parametric and Additive Uncertainty:\\ A Novel Constraint Tightening Approach}
% \title{Robust MPC for LPV Systems}

\author{Monimoy Bujarbaruah, Ugo Rosolia$^{\star}$, Yvonne R. St{\"u}rz$^{\star}$, Xiaojing Zhang,\\ and Francesco Borrelli\blfootnote{$^\star$authors contributed equally to this work. Emails:\{monimoyb, y.stuerz, fborrelli\}@berkeley.edu, urosolia@caltech.edu, xgeorge.zhang@gmail.com.}
}

\maketitle

\begin{abstract}
We propose a novel approach to design a robust Model Predictive Controller (MPC) for constrained uncertain linear systems. The uncertain system is modeled as linear parameter varying with additive disturbance. 
% The system dynamics matrices are not known exactly, leading to parametric model mismatch. We also consider the presence of an additive disturbance. 
% The uncertainty is modeled as an additive disturbance and a parametric n additive error on the system dynamics matrices. 
Set bounds for the system matrices and the additive uncertainty 
% each component of the model uncertainty 
are assumed to be known. We formulate 
% propose 
a novel optimization-based constraint tightening strategy around a predicted nominal trajectory which utilizes these bounds. 
% The resulting MPC controller guarantees robust satisfaction of state and input constraints in closed-loop with the uncertain system.
% for all possible realizations of the uncertainty, 
% while 
% % also mitigating such perturbations and 
% % steering the system to 
% avoiding restrictive constraint tightenings around the optimal predicted nominal trajectory. 
% This lowers the conservatism of the control synthesis problem.
With an appropriately designed terminal cost function and constraint set, we prove robust satisfaction of the imposed constraints by the resulting MPC in closed-loop with the uncertain system, and Input to State Stability of the origin. We highlight the efficacy of our proposed approach via a numerical example. 
% and demonstrate that due to the novel constraint tightening approach, the controller can lower its conservatism over a classical Shrinking Tube MPC strategy, and also a recent System Level Synthesis based approach. 
\end{abstract}

%%%%%%%%%%%%%%%%%%%%%%%%%%%%%%%%
\section{Introduction}
Model Predictive Control (MPC) is a well established optimal control strategy that is able to handle imposed constraints on system states and inputs \cite{mayne2000constrained}. The MPC approach is based on solving a constrained finite horizon optimal control problem at each time step 
% by minimizing a suitably chosen cost function involving the predicted states and a \emph{sequence} of predicted inputs, 
and applying the first optimal input to the plant. 
% The process is then repeated. 
% The predicted states along the MPC horizon are obtained after applying the predicted input sequence to a prediction model. 
% Under no discrepancy between the predictions from the prediction model and the evolution of the actual plant, 
% % being controlled,
% the feasibility of the MPC optimization problem
% % (and thus constraint satisfaction) 
% can be guaranteed at all time steps \cite{findeisen2002introduction, allgower2012nonlinear, kouvaritakis2016model, borrelli2017predictive}. 
% with so-called \emph{terminal conditions}. 
% For both linear and nonlinear systems in absence of any uncertainty in the prediction model, constructing these terminal sets is well studied in literature. 
A key challenge in MPC design is to guarantee robust constraint satisfaction in the presence of uncertainty in the prediction model.  
%Such uncertainty includes $(i)$ modeling mismatch between the prediction model and the actual plant, and $(ii)$ exogenous disturbances (e.g., process noise). 

% For uncertain nonlinear systems, although MPC control design is possible \cite{yu2013tube, faulwasser2018economic, kohler2018novel,  kohler2020computationally} for robust satisfaction of imposed constraints, the control design problem often faces non-convexity. 
% In order to restrict our attention on convex synthesis of MPC controllers, we specifically focus on Linear Time Invariant (LTI) systems in this paper. 

% For uncertain Linear Time Invariant (LTI) systems, numerous works in literature have been proposed to obtain the guarantees of recursive constraint satisfaction and stability of MPC controllers \cite{mayne2016robust}. 
% in closed-loop. 
% in presence of model uncertainty. 
For uncertain linear systems in presence of \emph{only} an additive disturbance, finding the optimal policy is NP-hard and typically involves dynamic programming \cite[Chapter~15]{borrelli2017predictive},\cite[Chapter~3]{rawlings2009model}, or Min-Max feedback \cite{scokaert1998min,bemporad2003min} approaches. Computationally tractable suboptimal robust MPC techniques such as tube MPC \cite{chisci2001systems,Goulart2006,rakovic2012parameterized} are well understood and widely used. The key idea is to restrict the input policy to affine or piecewise affine state feedback policies and then tighten the state constraints so that all trajectories within a ``tube" satisfy the imposed state constraints for all possible disturbances.

On the other hand, robust MPC design for uncertain linear systems in presence of \emph{both} a mismatch in the system dynamics matrices and an additive disturbance is more involved and is a topic of active research \cite{slsmpc,alonso2020distributed}. Min-Max MPC strategies could be computed in this case, but their computational complexity scales exponentially with the prediction horizon. Restricting the input policy parametrization to affine state feedback policies leads to  computationally tractable ellipsoidal regions of attraction (ROA)~\cite{blanchini1999set}. Such methods are presented in \cite{kothare1996robust,kouvaritakis2000efficient}. Polytopic, homothetic and elastic tube MPC methods with affine or piecewise affine state feedback policy parametrizations are introduced in \cite{langson2004robust,munoz2013recursively,rakovic2016elastic} to address the conservatism inherent to ellipsoidal ROA based methods, such as \cite{kothare1996robust,kouvaritakis2000efficient}. But the online computational complexity of these methods can noticeably increase
% over shrinking/fixed radius tube MPC \cite[Chapter~4]{kouvaritakis2016model}, 
while lowering conservatism, as shown in \cite{slsmpc} and \cite[Chapter~5]{kouvaritakis2016model}.
%with detailed analysis. 
% using the state and input constraint sets. 
% In order to design a classical shrinking or fixed radius tube MPC \cite[Chapter~3]{kouvaritakis2016model} using affine state feedback policy parametrization in presence of mismatch in the system matrices, the contribution of uncertainty due to the mismatches can be lumped together with the additive disturbance. A worst-case bound for this quantity can be found and then a method such as \cite{Goulart2006} can be used. Although computationally attractive,  
% % in practice 
% this 
% % is also highly
% % % , and often unnecessarily 
% % ``conservative". Such conservatism refers to
% results in highly conservative  constraint tightenings around a predicted nominal trajectory, as discussed in \cite[Section~3.2]{dean2018safely}. Consequently, similar to \cite{kothare1996robust, kouvaritakis2000efficient, schuurmans2000robust}, the feasible domain of the control synthesis problem shrinks. 
Alternatively, the work \cite{slsmpc,dean2018safely} utilizes a System Level Synthesis \cite{anderson2019system} based approach which obtains robust satisfaction of the imposed constraints with lower conservatism compared to \cite{kothare1996robust,kouvaritakis2000efficient}. 
% such aforementioned worst-case tightenings resulting from a ``net additive uncertainty" formulation. 
The approach can also be computationally more efficient than methods such as \cite{langson2004robust,rakovic2013homothetic,rakovic2016elastic}, as demonstrated in \cite{slsmpc}.
% solving the problem. 
% Furthermore, ellipsoidal region of attraction synthesis based methods are presented in \cite{kothare1996robust, kouvaritakis2000efficient, schuurmans2000robust}, but they also remain conservative, as pointed out in \cite{mayne2000constrained, kouvaritakis2016model}. To lower such conservatism in the control design, Polytopic and Homothetic Tube MPC methods are introduced in \cite{ rakovic2012homothetic, fleming2013regions, munoz2013recursively}. However, the online computational complexity of such methods noticeably increases over classical Shrinking/Fixed Radius Tube MPC, as pointed out in \cite{slsmpc} with detailed analysis. 

Motivated by the work of \cite{slsmpc,dean2018safely}, in this paper we propose a novel robust MPC approach for linear systems that can handle the presence of both a mismatch in the system matrices and an additive disturbance. Instead of using the worst-case constraint tightening tubes around any predicted nominal trajectory, we propose an optimization-based constraint tightening strategy which is a function of decision variables in the control synthesis problem, similar to \cite{lee2000linear,langson2004robust,rakovic2012homothetic,rakovic2012parameterized,evans2012robust,rakovic2013homothetic,munoz2013recursively,fleming2014robust,rakovic2016elastic}. 
% This lowers the conservatism in the proposed control design approach, while keeping it computationally viable for online control synthesis. 
Our contributions are summarized as:
\begin{itemize}
    \item We propose a novel constraint tightening strategy which is decoupled into two phases. In the first phase, we bound the effect of model uncertainty 
    % compute bounds on the effect of the mismatch in the system matrices 
    on any predicted \emph{nominal} (i.e., uncertainty free) trajectory. These bounds are computed offline. This phase is motivated by \cite{dean2018safely,slsmpc}. 
    % involves finding the solutions of a set of non-convex optimization problems offline. 
    % optimal solutions 
    % via function evaluations, 
    % without resorting to a nonlinear optimization solver. 
    In the second phase, the MPC is designed utilizing the above bounds, 
    % to guarantee robust satisfaction of imposed state and input constraints. 
    so that the constraint tightenings are functions of decision variables in the control synthesis problem. 
    This second phase is motivated by tube MPC works such as \cite{lee2000linear,langson2004robust,rakovic2012homothetic,rakovic2012parameterized,rakovic2013homothetic,munoz2013recursively,fleming2014robust,rakovic2016elastic}. 
    % it avoids restrictive constraint tightenings around the optimal predicted nominal trajectory.
    % , thus lowering conservatism of the control synthesis problem. 
    % Our strategy avoids parameter grid search methods, as used in \cite{slsmpc, dean2018safely}.
    
    \item We solve a tractable convex optimization problem online using a shrinking horizon approach for the MPC. With an appropriately constructed terminal set and a terminal cost, we prove robust satisfaction of the imposed constraints by the closed-loop system, and Input to State stability of the origin.
    % , adding to the guarantees of the MPC algorithm of \cite{slsmpc}. 
    
    \item With numerical simulations, we compare our proposed MPC approach with the tube MPC of \cite[Section~5]{langson2004robust} and also 
    % We call this approach the Uncertainty Feedback Robust MPC (UF-Robust MPC). 
    % We also compare our OB-Tube MPC approach 
    with the constrained LQR algorithm of \cite{dean2018safely}. 
    In the first case, we obtain at least 3x and up to 20x speedup in online control computations, with an approximately $4\%$ larger ROA by volume. In the latter case, we obtain an approximately 12x larger ROA by volume even with a safe open-loop policy.
    % larger region of attraction and lower online computational cost of our control synthesis problem for the considered scenarios.
    % by comparing the sets of feasible initial states.  
\end{itemize}
%%%%%%%%%%%%%%%%%%%%%
% The paper is organized as follows: In Section~\ref{mot_sec} we formulate the constrained optimal control problem, after introducing the system dynamics, and the state and input constraints. The novel constraint tightening is introduced and the robust MPC problem is presented in Section~\ref{sec:mpc}.  Section~\ref{sec:feas_n_stab} proves the feasibility and stability properties of the proposed robust MPC algorithm. Finally, numerical simulation results are presented in Section~\ref{sec:numerics}.   
%%%%%%%%%%%%%%%%%%%%%%%%%%%%%%%%%%%%%%%%%%%%%%%
\subsection*{Notation}
We use $\Vert \cdot \Vert$ to denote the norm of a vector. The dual norm of any vector norm $\Vert x \Vert$ for a vector $x$ is defined as $\Vert x \Vert_* = \sup_{\Vert v \Vert \leq 1}(v^\top x)$. The induced $p$-norm of any matrix $A$ is given by $\Vert A\Vert_p = \sup_{x \neq 0} \frac{\Vert Ax\Vert_p}{\Vert x\Vert_p}$, where $\Vert \cdot \Vert_p$ is the $p$-norm of a vector. The operation $A \otimes B$ denotes the Kronecker product of the matrices $A$ and $B$, and $\mathcal{A} \oplus \mathcal{B}$ denotes the Minkowski sum of the two sets $\mathcal{A}$ and $\mathcal{B}$. The set $K\mathcal{B}$ denotes the set of elements obtained from multiplying each element in the set $\mathcal{B}$ with $K$, i.e., $K \mathcal{B} = \{x: x = bK, b \in \mathcal{B}\}$. A continuous function $\alpha:[0,a) \rightarrow [0, \infty)$ is called a class-$\mathcal{K}$ function if it is strictly increasing in its domain and if $\alpha(0) = 0$. The class-$\mathcal{K}$ function belongs to class-$\mathcal{K}_\infty$ if  $a = \infty$ and $\lim_{r \rightarrow \infty} \alpha(r) = \infty$. A continuous function $\beta:[0,a) \times [0,\infty) \mapsto [0,\infty)$ is called a class-$\mathcal{KL}$ function if for each fixed $s$, the function $\beta(r,s)$ belongs to class-$\mathcal{K}$, and for each fixed $r$, $(i)$ the function $\beta(r,s)$ is decreasing w.r.t.\ $s$ and $(ii)$ $\beta(r,s) \rightarrow 0$ for $s \rightarrow \infty$. A real valued function $\alpha:[a,b] \mapsto \mathbb{R}$ is called Lipschitz with a Lipschitz constant $L$, if for all $x,y \in [a,b]$, we have $\Vert\alpha(x) - \alpha(y)\Vert \leq L\Vert x-y\Vert $. The sign $u \geq v$ between two vectors $u,v$ denotes element-wise inequality. $\mathrm{conv}(X,Y,\dots, Z)$ denotes the set of matrices that can be written as a convex combination of the matrices $X,Y,\dots,Z$. Notation $I_n$ is used to denote an identity matrix of dimension $n$ and $1_n$ denotes a vector of ones of length $n$. The consistency property for any induced $p$-norm and vector $q$-norm is given by $\Vert Xy\Vert_q \leq \Vert X\Vert_p \Vert y\Vert_q$, for any $X \in \mathbb{R}^{d_1 \times d_2}$ and $y \in \mathbb{R}^{d_2}$. The submultiplicativity property for any induced $p$-norm is given by $\Vert X Y \Vert_p \leq \Vert X \Vert_p \Vert Y \Vert_p$. 
%%%%%%%%%%%%%%%%%%%%%%%%%%%%%%%%%%%%%%%%%%%%%%%%%%%
\section{Problem Formulation}\label{mot_sec}
We consider linear  system dynamics
\begin{equation}\label{eq:unc_system}
     x_{t+1} = A x_t + B u_t + w_t,~x_0 = x_S,
\end{equation}
where $x_t\in \mathbb{R}^{d}$ is the state at time step $t$, $u_t\in\mathbb{R}^{m}$ is the input, and $A$ and $B$ are system dynamics matrices of appropriate dimensions. We assume that $A$ and $B$ are unknown matrices with estimates $\bar{A}$ and $\bar{B}$ available for control design \cite{campi2002finite}. In particular we assume
% \begin{subequations}
\begin{align}\label{eq:matrix_errors}
    & A = \bar{A} + \Delta^\mathrm{tr}_A,~B = \bar{B} + \Delta^\mathrm{tr}_B,
\end{align}
% \end{subequations}
where the \emph{true} parametric uncertainty matrices $\Delta^\mathrm{tr}_A$ and $\Delta^\mathrm{tr}_B$ are unknown and belong to  convex and compact sets \begin{align}\label{err_in_sets_pol}
    & \Delta^\mathrm{tr}_A \in \mathcal{P}_A,~\Delta^\mathrm{tr}_B \in \mathcal{P}_B.
\end{align}
We further assume that  $\mathcal{P}_A$ and $\mathcal{P}_B$ are given by the convex hulls of known \emph{vertex} matrices $\{\Delta_A^{(1)}, \Delta_A^{(2)}, \dots, \Delta_A^{(n_a)}\}$ and $\{\Delta_B^{(1)}, \Delta_B^{(2)}, \dots, \Delta_B^{(n_b)}\}$, with fixed $n_a, n_b >0$:
\begin{subequations}\label{eq:pol_out_termset_cond}
\begin{align}
    & \mathcal{P}_A = \mathrm{conv}(\Delta_A^{(1)}, \Delta_A^{(2)}, \dots, \Delta_A^{(n_a)}),\\
    & \mathcal{P}_B = \mathrm{conv}(\Delta_B^{(1)}, \Delta_B^{(2)}, \dots, \Delta_B^{(n_b)}).
    % ~\textnormal{satisfying} \nonumber\\
    % & \Phi_{A,p} \subseteq \mathcal{P}_A,~ \textnormal{and } \Phi_{B,p} \subseteq \mathcal{P}_B,
\end{align}
\end{subequations}
System~\eqref{eq:unc_system} is also affected by a disturbance $w_t$ with a convex and compact support $\mathbb{W} \subset \mathbb{R}^{d}$, i.e., $w_t\in\mathbb{W},~\forall~t\geq0$.
%%%

\begin{remark}
The proposed framework in this paper is also valid for time varying $\Delta^\mathrm{tr}_A$ and $\Delta^\mathrm{tr}_B$ satisfying \eqref{eq:pol_out_termset_cond}. 
% We stick to the constant (i.e., time invariant) case of matrix uncertainty in \eqref{eq:matrix_errors} without loss of generality, for the  clarity of presentation.
\end{remark}

Let the MPC horizon be $N$. Let $x_{k|t}$ denote the predicted state at time step $k$ for any possible uncertainty realization, obtained by applying the predicted input policies $\{u_{t|t},u_{t+1|t}(\cdot),\dots,u_{k-1|t}(\cdot)\}$ to system~\eqref{eq:unc_system}, and $\{\bar{x}_{k|t}, \bar{u}_{k|t}\}$ with $\bar{u}_{k|t} = u_{k|t}(\bar{x}_{k|t})$ denote the nominal state and corresponding input respectively. We are interested in synthesizing a robust MPC for the uncertain linear system \eqref{eq:unc_system}, by repeatedly solving the following finite time optimal control problem:
\begin{subequations}\label{eq:generalized_InfOCP}
	\begin{align}
V^{\star}(x_t,& \mathcal{P}_A, \mathcal{P}_B) = \notag \\
		\displaystyle\min_{U_t(\cdot)} ~ & \displaystyle\sum\limits_{k = t}^{t+N-1} \ell \left( \bar{x}_{k|t}, u_{k|t}\left(\bar{x}_{k|t}\right) \right)+Q(\bar{x}_{t+N|t}) \label{eq:FTOCP_cost} \\
		\text{s.t.,} & ~~~ \bar{x}_{k+1|t} = \bar{A} \bar{x}_{k|t} + \bar{B} u_{k|t}(\bar{x}_{k|t}), \label{eq:FTOCP_nominal} \\
		& ~~~x_{k+1|t} = Ax_{k|t} + Bu_{k|t}(x_{k|t}) + w_{k|t}, \label{eq:FTOCP_trueModel} \\
		&~~~\textnormal{with } A = \bar{A} + \Delta_A, B = \bar{B} + \Delta_B, \label{eq:FTOCP_modelUncertanty}\\[1ex]
		&~~~ H^x x_{k|t} \leq h^x,
		H^u u_{k|t}(x_{k|t}) \leq h^u, \label{eq:FTOCP_constr} \\
		&~~~x_{t+N|t} \in \mathcal{X}_N, \label{FTOCP_termC}\\[1ex]
		&~~~\forall w_{k|t} \in \mathbb W,~\forall \Delta_A \in \mathcal{P}_A,~\forall \Delta_B \in \mathcal{P}_B, \label{eq:FTOCP_uncertanty} \\
		&~~~\forall k \in \{t,t+1,\dots,(t+N-1)\}, \nonumber\\
		&~~~x_{t|t}=\bar{x}_{t|t} = x_t \nonumber,
	\end{align}
\end{subequations}
with $U_t(\cdot) = \{u_{t|t},u_{t+1|t}(\cdot),\dots,u_{t+N-1|t}(\cdot)\}$, and applying the optimal MPC policy
\begin{align}\label{eq:mpc_pol_formulation}
    % u^{ \scalebox{0.5}{\mathrm{MPC}} }_t(x_t) = u^\star_t(x_t),
    u^\mathrm{MPC}_t(x_t) = u^\star_{t|t}(x_t),
\end{align}
to system \eqref{eq:unc_system} in closed-loop.
% , where $x_{k|t}$ is the predicted state at time step $k$ for any possible uncertainty realization, obtained by applying the predicted input policies $\{u_{t|t},u_{t+1|t}(\cdot),\dots,u_{k-1|t}(\cdot)\}$ to system~\eqref{eq:unc_system}, and $\{\bar{x}_{k|t}, \bar{u}_{k|t}\}$ with $\bar{u}_{k|t} = u_{k|t}(\bar{x}_{k|t})$ denote the nominal state and corresponding input respectively. 
Problem \eqref{eq:generalized_InfOCP} is carried over to the space of feedback policies, 
$u_i(x_i)$ which map the set of feasible initial states, subset of $\mathbb{R}^d$, to the set of feasible inputs, subset of $\mathbb{R}^m$. 
% : \mathbb{R}^{d} \ni x_i \mapsto u_i = u_i(x_i)\in\mathbb{R}^{m},~\forall i \geq 0$, \textcolor{blue}{The optimal policy maps the state constraints to available controls}, 
The objective is to minimize the cost associated with the nominal model~\eqref{eq:FTOCP_nominal}. The true model~\eqref{eq:FTOCP_trueModel} and the uncertainty description~\eqref{eq:FTOCP_modelUncertanty} are used to guarantee that the constraints~\eqref{eq:FTOCP_constr}-\eqref{FTOCP_termC} are satisfied for all uncertainty realizations in~\eqref{eq:FTOCP_uncertanty}, where $H^x \in \mathbb{R}^{r \times d}, h^x \in \mathbb{R}^r, H^u \in \mathbb{R}^{o \times m}$ and $h^u \in \mathbb{R}^o$ describe the polytopes of states and input constraints. Finally, $\ell( x, u) = x^\top P x + u^\top R u$ is the stage cost and  $Q(x) = x^\top P_N x$ is the terminal cost. Assumption~\ref{assump:stagecost}-\ref{assump: termcost} in Section~\ref{sec:feas_n_stab} detail assumptions on $P,R$ and $P_N$. There are three main challenges with solving \eqref{eq:generalized_InfOCP}:
\begin{enumerate}[(A)]
    \item  The state and input constraints are to be satisfied robustly under the presence of mismatch in the system dynamics matrices and disturbances. In other words, \eqref{eq:FTOCP_constr}-\eqref{FTOCP_termC}-\eqref{eq:FTOCP_uncertanty} need to be reformulated so that they can be fed to a numerical programming algorithm.
    % This can lead to either conservative methods \cite{kothare1996robust, schuurmans2000robust}, or computationally intractable ones \cite[Chapter~15]{borrelli2017predictive}, \cite{scokaert1998min}. 
    
    \item Optimizing over policies $\{u_0, u_1(\cdot), u_2(\cdot), \dots\}$ in \eqref{eq:generalized_InfOCP} involves an optimization over infinite dimensional function spaces. This, in general,  is not computationally tractable for constrained linear systems.
    
    \item The feasibility of constraints \eqref{eq:FTOCP_constr} is to be robustly guaranteed at all time steps $t \geq 0$, for all admissible $\Delta_A \in \mathcal{P}_A, \Delta_B \in \mathcal{P}_B$, and for all $w_t~\in~\mathbb{W}$, such that 
    \begin{align*}
    &H^x x_t \leq h^x, H^u u^\mathrm{MPC}_t(x_t) \leq h^u, ~ \forall w_t \in \mathbb{W}, \forall t \geq 0,
    \end{align*}
where $x_{t+1} = Ax_t + Bu^\mathrm{MPC}_t(x_t) + w_t$.
% , and
% \begin{align*}
%     \bar{x}_t \rightarrow 0,~\textnormal{as } t \rightarrow \infty, 
% \end{align*}
% with $\bar{x}_{t+1} = A\bar{x}_t + Bu^\star_t(\bar{x}_t)$. Here 
% $x_0 = x_S$. 
\end{enumerate}
As common in the MPC literature, in this paper challenge (B) is addressed by restricting the input policy to the class of affine state feedback policies. Challenge (C) is addressed by appropriately constructing the terminal conditions, i.e., terminal set $\mathcal{X}_N$ in \eqref{FTOCP_termC} and terminal cost $Q(\cdot)$ in \eqref{eq:FTOCP_cost}, and using a safe backup policy in case \eqref{eq:generalized_InfOCP} loses feasibility. 

% \subsection{Contribution Highlights}
Various works in literature \cite{kothare1996robust,kouvaritakis2000efficient,langson2004robust,evans2012robust,munoz2013recursively,rakovic2013homothetic,fleming2014robust,rakovic2016elastic,dean2018safely,slsmpc} have been proposed to tackle Challenge (A). 
Our approach fundamentally differs from the others, because $(i)$ we compute bounds required for constraint tightenings in \eqref{eq:FTOCP_constr}-\eqref{FTOCP_termC} in a computationally expensive way offline, and then $(ii)$ we solve computationally efficient convex optimization problems online. This can lead to a large region of attraction, while limiting the online computational burden, as shown by our simulations in Section~\ref{sec:numerics}.
\section{Robust  MPC Design}\label{sec:mpc}
% In this section...brief intro
In this section we present the steps of the proposed robust MPC design approach, which solves problem \eqref{eq:generalized_InfOCP} at every time step $t \geq 0$. 
\subsection{Predicted State Evolution}\label{ssec:pred_state}
We first denote the sequences of vectors: 
\begin{equation}\label{notations1}
\begin{aligned}
&\mathbf{u}_t = [{u}^\top_{t|t}, {u}^\top_{t+1|t}(\cdot), \dots, {u}^\top_{t+N-1|t}(\cdot)]^\top,\\
& \bar{\mathbf{x}}_t = [\bar{x}^\top_{t|t}, \bar{x}^\top_{t+1|t}, \dots, \bar{x}^\top_{t+N-1|t}]^\top. 
\end{aligned}
\end{equation}
In this section, we use the following two observations: First, keeping the nominal state trajectory
% along the prediction horizon, i.e., 
$\bar{\mathbf{x}}_t$ as a decision variable in the MPC problem \eqref{eq:generalized_InfOCP} maintains certain structure that can be exploited to bound the effect of model uncertainty
% compute bounds on the perturbations 
on a predicted nominal trajectory, similar to \cite{slsmpc,dean2018safely}. 
% due to the presence of mismatch in the system model.  
And second, the predicted nominal trajectory and its associated inputs 
% policy $U_t(\cdot)$ 
along the horizon 
are computed by reformulating \eqref{eq:generalized_InfOCP} and solving a robust optimization problem, similar to tube MPC approaches such as \cite{lee2000linear,langson2004robust,rakovic2012homothetic,rakovic2012parameterized,rakovic2013homothetic,munoz2013recursively,fleming2014robust,rakovic2016elastic}. We thus attempt to merge the benefits of both these ideas in this work. 
% Next we detail the proposed approach. 

%chosen optimally by solving %\eqref{eq:MPC_R_fin}.

% This avoids worst-case constraint tightenings from \eqref{eq:wc_system_addtive}
% % by the solver in a way 
% around the optimal nominal trajectory.
% for \eqref{eq:constraints_nominal} 
% along the horizon are optimized over, while mitigating such perturbations. 
% This circumvents the worst-case bounds of constraint tightenings given by \eqref{eq:wc_system_addtive} at any point in the  state-space, thus allowing further room for \eqref{eq:MPC_R_fin} to be feasible.  

Recall the nominal system dynamics from \eqref{eq:FTOCP_nominal} given as
% \begin{equation*}
    $\bar x_{t+1} = \bar{A} \bar x_t + \bar{B} \bar{u}_t$,
% \end{equation*}
with $\bar{u}_t = u_t(\bar{x}_t)$.
% , where $\bar{A}$ and $\bar{B}$ are nominal estimates of true and \emph{unknown} matrices $A$ and $B$, respectively, as defined in \eqref{eq:matrix_errors}. 
% As shown with the trend of the first three predicted states in the Appendix, 
% Denoting the feedback policies $[u_{t|t}, u_{t+1|t}(\cdot), \dots, u_{t+N-1|t}(\cdot)]$ as simply $[u_{t|t}, u_{t+1|t}, \dots, u_{t+N-1|t}]$ henceforth for the convenience of notations, 
Denote the vectors $\mathbf{x}_t, \mathbf{w}_t \in \mathbb{R}^{dN}$ and $\Delta \mathbf{u}_t \in \mathbb{R}^{mN}$ as:
\begin{equation}\label{eq:state_prop_bigg}
\begin{aligned}
& {\mathbf{x}}_t =  \begin{bmatrix}
     x_{t+1|t}^\top &
     x_{t+2|t}^\top &
    \dots &
     x_{t+N|t}^\top
    \end{bmatrix}^\top,\\
    & \mathbf{w}_t = \begin{bmatrix}
     w_{t|t}^\top &
     w_{t+1|t}^\top &
    \dots 
     w_{t+N-1|t}^\top
    \end{bmatrix}^\top,\\
    % &\mathbf{u}_t = \begin{bmatrix}
    %  u_{t|t}^\top &
    %  u_{t+1|t}^\top(\cdot) &
    % \dots &
    %  u_{t+N-1|t}^\top(\cdot)
    % \end{bmatrix}^\top, \\
    & \Delta \mathbf{u}_t = \begin{bmatrix}
    \Delta u_{t|t}^\top &
     \Delta u_{t+1|t}^\top(\cdot) &
     \dots & \Delta u_{t+N-1|t}^\top(\cdot)
    \end{bmatrix}^\top,  
\end{aligned}
\end{equation}
where $\Delta u_{k|t}(\cdot) = u_{k|t}(\cdot) - \bar{u}_{k|t}$ for $k \in \{t,t+1,\dots,t+N-1\}$. Using \eqref{notations1} and  \eqref{eq:state_prop_bigg}, we can write the state evolution along the prediction horizon as:
\begin{equation}\label{eq:state_propagation}
\begin{aligned}
  \mathbf{x}_t & = \mathbf{A}^x \bar{\mathbf{x}}_t + \mathbf{A}^u \mathbf{u}_t + \mathbf{A}^{\Delta u} \Delta \mathbf{u}_t
    + \mathbf{A}^w \mathbf{w}_t,
    \end{aligned}
\end{equation}
where $\mathbf{x}_t$ denotes the prediction of possible evolutions of the realized states\footnote{Note,~\eqref{eq:state_prop_bigg} implies \eqref{eq:state_propagation} is a compact state update equation.}, and in \eqref{eq:state_propagation} the predicted nominal states along the horizon, i.e., $\bar{\mathbf{x}}_t$ from \eqref{notations1} appears directly and  \emph{not} expressed in terms of $\{x_t, u_{t|t}, u_{t+1|t}(\cdot), \dots, u_{t+N-1|t}(\cdot)\}$, as in \cite{Goulart2006}. The prediction dynamics matrices $\mathbf{A}^x, \mathbf{A}^u, \mathbf{A}^{\Delta u}$ and $\mathbf{A}^w$ in \eqref{eq:state_propagation} depend on $\bar{B}, \Delta_A, \Delta_B$ and $(\bar{A} + \Delta_A), (\bar{A} + \Delta_A)^2, \dots, (\bar{A} + \Delta_A)^{N-1}$. 
% and are reported in the Appendix. 
We define $A_\Delta = \bar{A} + \Delta_A$ for some possible $\Delta_A \in \mathcal{P}_A$. Then $A_\Delta \in \mathcal{P}_{A_\Delta}$, with the set $\mathcal{P}_{A_\Delta}$ defined as:
\begin{align}\label{padelta}
&\mathcal{P}_{A_\Delta} = \{A_m: A_m = \bar{A} + \Delta_A,~\Delta_A \in \mathcal{P}_A\}.
\end{align}
% and denote $A_\Delta \in \mathcal{P}_{A_\Delta}$, where the set $\mathcal{P}_{A_\Delta}$ can be obtained from the set $\mathcal{P}_A$, as $\mathcal{P}_{A_\Delta} = \bar{A} \oplus \mathcal{P}_A$. 
Using \eqref{padelta} we rewrite the matrices in \eqref{eq:state_propagation} as follows:
\begin{equation}\label{simplified_dyn_matrices}
\begin{aligned}
    & \mathbf{A}^x 
     =  \bar{\mathbf{A}} + \Big (\bar{\mathbf{A}}_1 + \mathbf{A}_\delta \Big ) \mathbf{\Delta}_A,\\
     &\mathbf{A}^u =  \bar{\mathbf{B}} + \Big ( \bar{\mathbf{A}}_1 + \mathbf{A}_\delta \Big ) \mathbf{\Delta}_B,\\
    & \mathbf{A}^{\Delta u} = \Big ( \bar{\mathbf{A}}_1 - \mathbf{I}_d + \mathbf{A}_\delta  \Big) \bar{\mathbf{B}},~\textnormal{and},\\
    &\mathbf{A}^w  = \mathbf{I}_d + \bar{\mathbf{A}}_v \mathbf{A}_\Delta,
\end{aligned}
\end{equation}
where $\mathbf{I}_d = (I_N \otimes I_d) \in \mathbb{R}^{dN \times dN}, \bar{\mathbf{A}} = (I_N \otimes \bar{A}) \in \mathbb{R}^{dN \times dN}, \bar{\mathbf{B}} = (I_N \otimes \bar{B}) \in \mathbb{R}^{dN \times mN}, \mathbf{\Delta}_A = (I_N \otimes \Delta_A) \in \mathbb{R}^{dN \times dN}$, and $\mathbf{\Delta}_B = (I_N \otimes \Delta_B) \in \mathbb{R}^{dN \times mN}$. The matrices $\bar{\mathbf{A}}_1$, $\mathbf{A}_\delta$,  $\bar{\mathbf{A}}_v$ and $\mathbf{A}_\Delta$ are defined in \ref{A1} in the Appendix. Matrices $\mathbf{A}_\delta$ and $\mathbf{A}_\Delta$ depend on parametric uncertainty matrices  $\Delta_A$ and $\Delta_B$. In the next sections, we substitute the matrices from \eqref{simplified_dyn_matrices} in \eqref{eq:state_propagation} in order to design a control policy that robustly satisfies \eqref{eq:FTOCP_constr}-\eqref{FTOCP_termC} along the prediction horizon. 
%%%%%%%%%%%%%%%%%%%%%%%%%%%%%%%%%
\subsection{Novel Optimization-Based Constraint Tightening}\label{sec:opt_tight}
% We first bound the effect of model mismatch, i.e., the matrices $\mathbf{A}_\delta, \mathbf{A}_\Delta, \mathbf{\Delta}_A, \mathbf{\Delta}_B$ on predicted nominal states. These bounds, denoted as $\{\mathbf{t}^i_w, \mathbf{t}^i_1, \mathbf{t}^i_2, \mathbf{t}^i_3\}$ for rows $i \in \{1,2,\dots, r(N-1)+r_N\}$, are computed \emph{offline}, and are derived in detail in \eqref{eq:bound_mainterm}-\eqref{fourthbound} in the Appendix. 
The terminal set $\mathcal{X}_N$ in \eqref{FTOCP_termC} is defined by 
$\mathcal{X}_N = \{x: H^x_N x \leq h^x_N\}$, with $H^x_N \in \mathbb{R}^{r_N \times d}, h^x_N \in \mathbb{R}^{r_N}$. 
We denote the matrix $\mathbf{F}^x = \mathrm{diag}(I_{N-1} \otimes H^x, H^x_N) \in \mathbb{R}^{(r(N-1)+r_N) \times dN}$, $\mathbf{f}^x = (h^x, h^x, \dots, h_N^x) \in \mathbb{R}^{r(N-1)+r_N}$ for any given $N$. Using \eqref{eq:state_propagation}, the robust state constraints in \eqref{eq:generalized_InfOCP} for predicted states along the prediction horizon and at the end of the horizon can then be written as:
\begin{align}\label{ugo_wants_it_state}
    & \mathbf{F}^x \mathbf{x}_t \leq \mathbf{f}^x,~\forall \Delta_A \in \mathcal{P}_A,~\forall \Delta_B \in \mathcal{P}_B,~\forall w_t \in \mathbb{W}.
\end{align}
We guarantee satisfaction of \eqref{ugo_wants_it_state} using the following: Suppose for any $a,b$, we need to guarantee $a \leq b$. We first obtain an upper bound $c$, such that $a\leq c$, and then we impose $c \leq b$. This is a sufficient condition for $a \leq b$. Accordingly, using \eqref{eq:state_propagation} and \eqref{simplified_dyn_matrices}
% The bounds \eqref{eq:bound_mainterm}-\eqref{fourthbound} can be computed \emph{offline} and 
% the  robust state constraints  
constraint \eqref{ugo_wants_it_state} for all time steps $t \geq 0$ can be replaced row-wise as:
\begin{align}\label{eq:fin_state_con}
    & \mathbf{F}^x_i((\bar{\mathbf{A}} + \bar{\mathbf{A}}_1 \mathbf{\Delta}_A) \bar{\mathbf{x}}_t + (\bar{\mathbf{B}} + \bar{\mathbf{A}}_1 \mathbf{\Delta}_B) \mathbf{u}_t + (\bar{\mathbf{A}}_1 - \mathbf{I}_d)\bar{\mathbf{B}} \Delta \mathbf{u}_t + \mathbf{w}_t) + \mathbf{t}^i_1 \Vert \bar{\mathbf{x}}_t \Vert + \mathbf{t}^i_2 \Vert \mathbf{u}_t \Vert + \cdots \nonumber \\
    & ~~~~~~~~~~~~~~~~~~~~~~~~~~~~~~~~~~~~~~~~~~~~~~~~~~~~~~~~~~~~~~~~~~~~~~~~~ + \mathbf{t}^i_3 \Vert \Delta \mathbf{u}_t \Vert + \mathbf{t}^i_w \Vert \mathbf{w}_t \Vert \leq \mathbf{f}^x_i, \nonumber\\[1ex]
    & \forall \Delta_A \in \mathcal{P}_A,~\forall \Delta_B \in \mathcal{P}_B,~\forall w_t \in \mathbb{W},
\end{align}
for $i \in \{1,2,\dots,r(N-1) + r_N\}$, where recall that $r$ and $r_N$ are the number of rows of $H^x$ and $H^x_N$, respectively. In Appendix~\ref{GIVEMEP} we detail the derivation of \eqref{eq:fin_state_con} from \eqref{ugo_wants_it_state} and the computation of the bounds $\{\mathbf{t}^i_w, \mathbf{t}^i_1, \mathbf{t}^i_2, \mathbf{t}^i_3\}$ for rows $i \in \{1,2,\dots, r(N-1)+r_N\}$. In \eqref{eq:fin_state_con} we have bounded the effect of model mismatch, i.e., the matrices $\mathbf{A}_\delta, \mathbf{A}_\Delta, \mathbf{\Delta}_A, \mathbf{\Delta}_B$ on predicted nominal states. These bounds, denoted as $\{\mathbf{t}^i_w, \mathbf{t}^i_1, \mathbf{t}^i_2, \mathbf{t}^i_3\}$ for rows $i \in \{1,2,\dots, r(N-1)+r_N\}$, are computed \emph{offline}, and are derived in detail in \eqref{eq:bound_mainterm}-\eqref{fourthbound} in the Appendix, where we also show that~\eqref{eq:fin_state_con} is sufficient for~\eqref{ugo_wants_it_state}. 

In constraint \eqref{eq:fin_state_con}, note that the decision variables are the predicted nominal trajectory $\bar{\mathbf{x}}_t$, and the sequence of input policies $\mathbf{u}_t$. These decision variables multiply effects of the bounds $\mathbf{t}^i_1, \mathbf{t}^i_2$ and $\mathbf{t}^i_3$. In conclusion, the tightening
of the original constraint \eqref{eq:FTOCP_constr} 
proposed in~\eqref{eq:fin_state_con} depends on the optimization variables, $\bar{\mathbf{x}}_t$, $\mathbf{u}_t$, and $\Delta \mathbf{u}_t$. This is a key contribution of our work. 
% unlike classical shrinking and fixed radius tube MPC \cite{chisci2001systems, Goulart2006, borrelli2017predictive, kouvaritakis2016model}. In these methods the constraint tightening is based on the worst-case value of an \emph{additive} uncertainty, resulting in approaches such as \eqref{eq:MPC_R_fin_add} for solving \eqref{eq:MPC_R_fin}. 
Alternatively in \cite{dean2018safely,slsmpc}, the constraint tightening is obtained bounding the closed-loop system response, which involves the norm of the product between the decision variables and the uncertainty. 
% Thus, bounds such as $\{\mathbf{t}^i_w, \mathbf{t}^i_1, \mathbf{t}^i_2, \mathbf{t}^i_3\}$ (derived in \eqref{eq:bound_mainterm}-\eqref{fourthbound} in the Appendix) which are decoupled from decision variables are not obtained. 
Therefore the method in \cite{dean2018safely,slsmpc} needs to resort to a grid search over parameters to obtain sufficient conditions for satisfying \eqref{eq:FTOCP_constr} robustly. Tube MPC methods such as \cite{langson2004robust,munoz2013recursively,rakovic2013homothetic,fleming2014robust,rakovic2016elastic}, summarized in \cite[Chapter~5]{kouvaritakis2016model}, could lead to tightenings equivalent to \eqref{eq:fin_state_con} under appropriately chosen parametrization of tube cross sections. However, such parametrizations aren't immediate. 

% We provide comparisons of our algorithm with one of each of the aforementioned approaches in Section~\ref{sec:numerics}.

%%%%%%%%%%%%%%%%%%%%%%%%%%%%%%%%%%%%%%%%%%%%%
\subsection{Control Policy Parametrization}\label{sec:pol_par}
Recall Challenge (B) mentioned in Section~\ref{mot_sec}. To address this, we restrict ourselves to the affine disturbance feedback parametrization \cite{lofberg2003minimax,Goulart2006} for MPC control synthesis. For all predicted steps $k \in \{t,t+1,\dots,t+N-1\}$ over the MPC horizon, the control policy is chosen as:
\begin{equation}\label{eq:inputParam_DF_OL}
	u_{k|t}(x_{k|t}) = \sum \limits_{l=t}^{k-1}M_{k,l|t} w_{l|t}  + \bar{u}_{k|t},
\end{equation}
where $M_{k|t}$ are the \emph{planned} feedback gains at time step $t$ and $\bar{u}_{k|t} = u_{k|t}(\bar{x}_{k|t})$ are the auxiliary nominal inputs. Then the sequence of predicted inputs from \eqref{eq:inputParam_DF_OL} can be written as $\mathbf{u}_t = \mathbf{M}^{(N)}_t\mathbf{w}_t + \mathbf{\bar{u}}^{(N)}_t$ at time step $t$, where $\mathbf{M}^{(N)}_t\in \mathbb{R}^{mN \times dN}$ and $\mathbf{\bar{u}}^{(N)}_t \in \mathbb{R}^{mN}$ are
\begin{equation*}
\begin{aligned}
  &\mathbf{M}^{(N)}_t  =  \begin{bmatrix}0& \dots&\dots&0\\
  M_{t+1,t}& 0 & \dots & 0\\
  \vdots &\ddots& \ddots &\vdots\\
  M_{t+N-1,t}& \dots& M_{t+N-1,t+N-2}& 0
  \end{bmatrix},\\
  & \mathbf{\bar{u}}^{(N)}_t = [\bar{u}_{t|t}^\top,\bar{u}_{t+1|t}^\top , \dots, \bar{u}_{t+N-1|t}^\top]^\top.
\end{aligned}
\end{equation*}
% Note that given the policy parametrization \eqref{eq:inputParam_DF_OL} for system \eqref{eq:unc_system}, one can find an equivalent set of state feedback gains $K_{k|t} \in \mathbb{R}^{m \times d}$ and auxiliary inputs $v_{k|t} \in \mathbb{R}^{m}$, such that
% $u_{k|t}(x_{k|t}) = K_{k|t}x_{k|t} + v_{k|t}$, for $k \in \{t,t+1,\dots,t+N-1\}$. See \cite{lofberg2003minimax,Goulart2006} for further details on this equivalence. 
%%%%%%%%%%%%%%%%%%%%%%%%%%%%%%%%%%%%%
\subsection{Terminal Set Construction}\label{ssec:term_set}
The terminal set $\mathcal{X}_N$ is designed in this section to address Challenge (C) mentioned in Section~\ref{sec:mpc}. In particular, the terminal set $\mathcal{X}_N$ is chosen as the maximal robust positive invariant set of an autonomous system under a linear feedback policy, chosen as
\begin{align}\label{eq:term_pol}
    \kappa_N(x) = Kx,%~\forall x \in \mathcal{X}_N,
\end{align}
where $K \in \mathbb{R}^{m \times d}$ is the feedback gain. Recall the sets $\mathcal{P}_A$ and $\mathcal{P}_B$ from \eqref{eq:pol_out_termset_cond} and $\mathcal{P}_{A_\Delta}$ from \eqref{padelta}. Now consider 
% the set $\mathcal{P}_{B_\Delta}$ defined as
\begin{align*}
    % & \mathcal{P}_{A_\Delta} = \{A_m: A_m = \bar{A} + \Delta_A,~\forall \Delta_A \in \mathcal{P}_A\},\\
    & \mathcal{P}_{B_\Delta} = \{B_m: B_m = \bar{B} + \Delta_B,~\Delta_B \in \mathcal{P}_B\}. 
\end{align*}
Under policy \eqref{eq:term_pol}, the  closed-loop system dynamics matrix considered for constructing the terminal set satisfies
\begin{align*}
A^\mathrm{cl}=A + BK \in \mathcal{P}_{A_\Delta} \oplus K\mathcal{P}_{B_\Delta}. 
\end{align*}
The following assumption guarantees that $K$ robustly stabilizes the system and analogous assumptions are common in robust MPC literature
\cite{kothare1996robust,langson2004robust,rakovic2013homothetic,fleming2014robust,munoz2013recursively, vicente2019stabilizing}. 

\begin{assumption}\label{assump:stable}
$A^\mathrm{cl}_m = (A_m + B_mK)$ is stable for all $A_m \in \mathcal{P}_{A_\Delta}$ and $B_m \in \mathcal{P}_{B_\Delta}$. 
\end{assumption}
% Assumption~\ref{assump:stable} can be satisfied by synthesizing the stabilizing gain $K$ following standard robust MPC techniques, or methods such as \cite{soh1990schur}. 
% A further discussion on the requirement of Assumption~\ref{assump:stable} is included in  Section~\ref{sec:feas_n_stab}, Remark~\ref{rem:assump1}. 
% We denote the set of all closed loop dynamic matrices as:
% \begin{align*}
%     \mathcal{P}_{A_\Delta}^\mathrm{cl} = \{A^\mathrm{cl}_m: A^\mathrm{cl}_m = A_m + B_m K,~\forall A_m \in \mathcal{P}_{A_\Delta},~\forall B_m \in \mathcal{P}_{B_\Delta}\}.
% \end{align*}
The gain $K$ satisfying Assumption~\ref{assump:stable} can be chosen by following a method such as \cite{kothare1996robust, boyd1994linear}. Using Assumption~\ref{assump:stable}, set ${\mathcal{X}}_N$ can then be computed as the maximal robust positive invariant set for the autonomous dynamics 
\begin{align}\label{eq:aut_sys_term}
    x_{t+1} = (A_m+B_mK) x_t + w_t,
\end{align}
for all $A_m \in \mathcal{P}_{A_\Delta}, B_m \in \mathcal{P}_{B_\Delta}$, and $w_t \in \mathbb{W}$. That is,
\begin{equation}\label{eq:term_set_DF}
    \begin{aligned}
    &\mathcal{X}_N \subseteq \{x|H^x x \leq h^x,~H^uKx \leq h^u\},\\
    &(A_m + B_mK)x + w \in \mathcal{X}_N,~\\
    &\forall x\in \mathcal{X}_N,~\forall A_m \in \mathcal{P}_{A_\Delta},~\forall B_m \in \mathcal{P}_{B_\Delta},~\forall w \in \mathbb{W}.
    \end{aligned}
\end{equation}
See \cite[Section~10.3.3]{borrelli2017predictive} for a fixed point iteration algorithm used to compute $\mathcal{X}_N$. This algorithm has no convergence guarantees  \cite{vidal1999controlled}.

\subsection{Tractable MPC Problem with Safe Backup}\label{ssec:mpc_problem}
In this section we present the MPC reformulation of \eqref{eq:generalized_InfOCP} which guarantees robust constraint satisfaction at all time steps $t \geq 0$, and Input to State Stability of the origin.
%of the  the reformulated robust state constraints \eqref{eq:fin_state_con}.
%with input policy \eqref{eq:inputParam_DF_OL}, and the terminal set $\mathcal{X}_N$ from Algorithm~\ref{alg1}. 
We start with the following observation: The terminal set $\mathcal{X}_N$ from \eqref{eq:term_set_DF} is robustly invariant to all uncertainty of the form: $\forall \Delta_A \in \mathcal{P}_A,~\forall \Delta_B \in \mathcal{P}_B,~\forall w \in \mathbb{W},~\forall t \geq 0$, when the state feedback policy $\kappa_N(x) = Kx$ is used in \eqref{eq:unc_system}. However, along the prediction horizon we use bounds $\{\mathbf{t}^i_w, \mathbf{t}^i_1, \mathbf{t}^i_2, \mathbf{t}^i_3\}$, which are obtained by more conservative tightenings from H{\"o}lder's and triangle inequalities, and induced norm consistency and submultiplicativity properties (see \eqref{eq:bound_mainterm}-\eqref{fourthbound} in the Appendix). Thus the uncertainty bounds along the horizon over-approximate the effect of the true uncertainty used to compute the  terminal set. 
% However, the above is \emph{not} true along the prediction horizon, where we use bounds $\{\mathbf{t}^i_w, \mathbf{t}^i_1, \mathbf{t}^i_2, \mathbf{t}^i_3\}$ which are obtained by more conservative tightenings from H{\"o}lder's and triangle inequalities, and induced norm consistency and submultiplicativity properties (see \eqref{eq:bound_mainterm}-\eqref{fourthbound} in the Appendix). 
% Thus the uncertainty bounds will be tighter along the horizon compared to the bounds used for the terminal set.
This implies that the classical shifting argument \cite[Chapter~12]{borrelli2017predictive} for recursive MPC feasibility cannot be used. As a consequence, to ensure robust satisfaction of constraints~\eqref{eq:FTOCP_constr} by system~\eqref{eq:unc_system} at all time steps and Input to State Stability of the origin,  we will use the following strategy: $(i)$ at any given time step, we solve the MPC reformulation of problem \eqref{eq:generalized_InfOCP} in a shrinking horizon fashion, i.e., we choose the MPC horizon length at time step $t$, denoted by $N_t$, as:
\begin{align}\label{eq:n_t}
N_t = \begin{cases} N-t,~ \textnormal{if } t \in \{0,1,\dots N-2\}, \\ 1,~\textnormal{otherwise}. \end{cases}
\end{align}
If the shrinking horizon MPC problem is infeasible, we use the time-shifted optimal policy from a previous time step as a safe backup policy to guarantee robust satisfaction of \eqref{eq:FTOCP_constr}, and $(ii)$  we design the terminal cost matrix $P_N$ so that the MPC open-loop cost is a Lyapunov function inside $\mathcal{X}_N$. This design choice, together with the shrinking horizon strategy, which guarantees finite time convergence to $\mathcal{X}_N$, allows us to show Input to State Stability of the origin.
% we use the adaptive horizon strategy as used in \cite{langson2004robust,rosolia2019robust}. At any given time step $t$, we solve a set of $N$ convex optimization problems for control synthesis, with the prediction horizon
% % now denoted as $N_t$ varying from $1$ to $N$, i.e., 
% $N_t \in \{1,2,\dots,N\}$. 
% If at least one of these $N$ problems is feasible at time step $0$, we guarantee feasibility of at least one of the $N$ problems at any time step $t>0$. This is proven in detail in Section~\ref{sec:feas_n_stab}, Theorem~\ref{thm1}. The number of optimization problems solved can be reduced from $N$ to two, while maintaining all the guarantees of Section~\ref{sec:feas_n_stab} (see Section~\ref{sec:disc} for further details).

We introduce the following set of required notations. 
% Let $N_t$ denote the MPC horizon length at time step $t$. We choose:
% \begin{align}\label{eq:n_t}
% N_t = \begin{cases} N-t,~ \textnormal{if } t \in \{0,1,\dots N-2\}, \\ 1,~\textnormal{otherwise}. \end{cases}
% \end{align}
Denote the set $\mathbb{W} = \{w \in \mathbb{R}^d: H^w w \leq h^w\}$ with $H^w \in \mathbb{R}^{a \times d}$ and $h^w \in \mathbb{R}^{a}$. For a horizon length of $N_t$ from \eqref{eq:n_t}, this gives $\mathbf{W} = \{\mathbf{w} \in \mathbb{R}^{dN_t}: \mathbf{H}^w \mathbf{w} \leq \mathbf{h}^w\}$, with $\mathbf{H}^w = I_{N_t} \otimes H^w \in \mathbb{R}^{aN_t \times dN_t}$ and $\mathbf{h}^w = (h^w, h^w, \dots, h^w) \in \mathbb{R}^{aN_t}$. Also denote the matrices $\mathbf{H}^u = I_{N_t} \otimes H^u \in \mathbb{R}^{oN_t \times m N_t}$, and $\mathbf{h}^u = ({h}^u, {h}^u, \dots, {h}^u) \in \mathbb{R}^{o N_t}$. Moreover, we denote vectors $\mathbf{t}^{(N_t)}_j = [\mathbf{t}^1_j, \mathbf{t}^2_j, \dots, \mathbf{t}^{r(N_t-1)+r_N}_j ]^\top$ for the indices $j \in \{w,1,2,3\}$. 
% For notational convenience let us define
% \begin{subequations}\label{eq:lumpe\tilde{w}_t}
% \begin{align}
%     & \mathbf{t}^{(N_t)}_{\delta 1} = \mathbf{t}^{(N_t)}_{\delta A} + \mathbf{t}^{(N_t)}_1,\\
%     & \mathbf{t}^{(N_t)}_{\delta 2} = \mathbf{t}^{(N_t)}_{\delta B} + \mathbf{t}^{(N_t)}_2,\\
%     & \mathbf{t}^{N_t}_{\delta 3} = \mathbf{t}^{(N_t)}_{\delta B} + \mathbf{t}^{(N_t)}_2 + \mathbf{t}^{(N_t)}_3,
% \end{align}
% \end{subequations}
% for each $N_t \in \{1,2,\dots, N\}$. 
We use the notation $\bar{\mathbf{x}}^{(N_t)}_t$ for each horizon length $N_t$, to explicitly indicate the varying dimension of the vector $\bar{\mathbf{x}}_t$ previously introduced in \eqref{eq:state_prop_bigg}. 

In \eqref{eq:fin_state_con} the input policy was not specified. We now use policy parametrization \eqref{eq:inputParam_DF_OL} in \eqref{eq:fin_state_con} and consider the following two cases\footnote{The dimensions of $\mathbf{F}^x, \mathbf{f}^x, \bar{\mathbf{A}}, \bar{\mathbf{B}}, \bar{\mathbf{A}}_1, \mathbf{I}_d$ and $\mathbf{w}_t$ vary depending on $N_t$. We omit showing this explicitly for brevity.}: 
\begin{subequations}\label{eq:state_robcon}
\begin{align}
& \textnormal{\textbf{Case 1}: ($N_t \geq 2$, i.e., $t \leq  N-2$)} \nonumber \\
& \max_{\mathbf{w}_t \in \mathbf{W}}  \mathbf{F}^x \Big (  \bar{\mathbf{A}} \bar{\mathbf{x}}^{(N_t)}_t + \bar{\mathbf{B}} (\mathbf{M}^{(N_t)}_t \mathbf{w}_t + \bar{\mathbf{u}}^{(N_t)}_t) + (\bar{\mathbf{A}}_1 - \mathbf{I}_d)\bar{\mathbf{B}}\mathbf{M}^{(N_t)}_t \mathbf{w}_t + \mathbf{w}_t \Big ) \leq \mathbf{f}_\mathrm{tight}^x, \label{Ng1_con}\\
& \textnormal{\textbf{Case 2}: ($N_t = 1$, i.e., $t \geq N-1$)} \nonumber \\ 
& \max_{\substack{{w}_t \in \mathbb{W}\\ \Delta_A \in \mathcal{P}_A \\ \Delta_B \in \mathcal{P}_B}}  H^x_N((\bar{{A}}+{\Delta}_A) \bar{\mathbf{x}}^{(1)}_t + (\bar{{B}} + {\Delta}_B) \bar{\mathbf{u}}^{(1)}_t + {w}_t) \leq h^x_N, \label{N1_con}
\end{align}
\end{subequations}
for $N_t \in \{1,2,\dots, N\}$. The tightened set of constraints  $\mathbf{f}^x_{\mathrm{tight}}$ are given by
\begin{align}\label{tight1}
% & f^x_\mathrm{tight} \nonumber \\
\!\mathbf{f}^x_\mathrm{tight}& \!=\! \mathbf{f}^x\! -\! \mathbf{t}^{(N_t)}_{\delta 1} \Vert \bar{\mathbf{x}}^{(N_t)}_t \Vert \!-\! \mathbf{t}^{(N_t)}_{\delta 3} \Vert \mathbf{M}^{(N_t)}_t \Vert_p \mathbf{w}_\mathrm{max} - \mathbf{t}^{(N_t)}_{\delta 2} \Vert \bar{\mathbf{u}}^{(N_t)}_t \Vert - \mathbf{t}^{(N_t)}_w \mathbf{w}_\mathrm{max},
\end{align}
with $\Vert \mathbf{w}_t \Vert \leq \mathbf{w}_\mathrm{max}$ for all $t \geq 0$, where
\begin{equation}\label{tdel1_etc}
\begin{aligned}
    & \mathbf{t}^{(N_t)}_{\delta 1} = \mathbf{t}^{(N_t)}_{\delta A} + \mathbf{t}^{(N_t)}_1,~\mathbf{t}^{(N_t)}_{\delta 2} = \mathbf{t}^{(N_t)}_{\delta B} + \mathbf{t}^{(N_t)}_2,\\
    & \mathbf{t}^{(N_t)}_{\delta 3} = \mathbf{t}^{(N_t)}_{\delta B} + \mathbf{t}^{(N_t)}_2 + \mathbf{t}^{(N_t)}_3,
\end{aligned}
\end{equation}
using the bounds
\begin{subequations}\label{eq:deladelb_bounds}
\begin{align}
    & \max_{\Delta_A \in \mathcal{P}_A} \Vert \mathbf{F}^x_i \bar{\mathbf{A}}_1 \mathbf{\Delta}_A \Vert_* 
    % & \leq  \Vert \mathbf{F}^x_i \bar{\mathbf{A}}_1 \Vert_*  \max_{\Delta_A \in \mathcal{P}_A} \Vert \mathbf{\Delta}_A \Vert_p, \nonumber\\
    = \mathbf{t}^{(N_t),i}_{\delta A},\\
    & \max_{\Delta_B \in \mathcal{P}_B} \Vert \mathbf{F}^x_i \bar{\mathbf{A}}_1 \mathbf{\Delta}_B \Vert_* 
    % & \leq  \Vert \mathbf{F}^x_i \bar{\mathbf{A}}_1 \Vert_*  \max_{\Delta_B \in \mathcal{P}_B} \Vert \mathbf{\Delta}_B \Vert_p, \nonumber \\
    = \mathbf{t}^{(N_t),i}_{\delta B},
\end{align}
\end{subequations}
for $i \in \{1,2,\dots, r(N_t-1)+r_N\}$. See \ref{Ap_bnd_der} in the Appendix for a derivation of \eqref{eq:state_robcon}-\eqref{tight1} from \eqref{eq:fin_state_con} using the bounds \eqref{tdel1_etc}. Having formulated the state constraints, the input constraints in \eqref{eq:FTOCP_constr} along the horizon can be written as:
\begin{align}\label{eq:input_robcon} 
    & \max_{\mathbf{w}_t \in \mathbf{W}}  \mathbf{H}^u \Big ( \mathbf{M}^{(N_t)}_t \mathbf{w}_t + \bar{\mathbf{u}}^{(N_t)}_t \Big) \leq \mathbf{h}^u, 
\end{align}
for $N_t \in \{1,2,\dots, N\}$. Using \eqref{eq:state_robcon}-\eqref{eq:input_robcon}, at any time step $t$ we then solve 
\begin{equation}\label{eq:MPC_R_fin_trac}
	\begin{aligned} 
% 	& \textnormal{\textbf{MPC Optimization Problem:}}\\[1ex]
	  V_{t \rightarrow t+N_t}^{\mathrm{MPC}}&(x_t, \mathbf{t}^{(N_t)}, N_t)   :=  \\
	& \min_{\substack{\mathbf{M}^{(N_t)}_t, \mathbf{z}_t^{(N_t)}}} ~ (\mathbf{z}_t^{(N_t)})^\top \bar Q^{(N_t)} \mathbf{z}_t^{(N_t)} \\
% 		&~~~~~\text{s.t., }~~~~~~~x_{k+1|t} = Ax_{k|t} + Bu_{k|t}(x_{k|t}) + w_{k|t}, \\
% 		&~~~~~~~~~~~~~~~~~~\textnormal{with } A = \bar{A} + \Delta_A, B = \bar{B} + \Delta_B,  \\[1ex]
		& ~~~~~\text{s.t., }~~~~~ G^{(N_t)}_\mathrm{eq}\mathbf{z}_t^{(N_t)} = b^{(N_t)}_\mathrm{eq} x_t,  \\
% 		& ~~~~~\text{s.t., }~~~~~~~ \bar{x}_{k+1|t} = \bar{A}\bar{x}_{k|t} + \bar{B}\bar{u}_{k|t},  \\
% 		& ~~~~~~~~~~~~~~~~~~u_{k|t} (x_{k|t})= \sum \limits_{l=t}^{k-1}M_{k,l|t} w_{l|t}  + \bar{u}_{k|t},  \\
	   %& ~~~~~~~~~~~~~~~ \max_{\mathbf{w}_t \in \mathbf{W}} \Bigg ( F^x_i(\bar{\mathbf{A}} \bar{\mathbf{x}}_t + \bar{\mathbf{B}} (\mathbf{M}_t \mathbf{w}_t + \bar{\mathbf{u}}_t) + \mathbf{w}_t) + (\mathbf{t}^i_{\delta A} + \mathbf{t}^i_1) \Vert \bar{\mathbf{x}}_t \Vert_\infty + ( \mathbf{t}^i_{\delta B} + \mathbf{t}^i_2) \Vert \mathbf{M}_t \Vert_1 \mathbf{w}_\mathrm{max} + \nonumber \\
	   %& ~~~~~~~~~~~~~~~~~~~~~~~~~~~~~~~~ \cdots + (\mathbf{t}^i_{\delta B} + \mathbf{t}^i_2) \Vert \bar{\mathbf{u}}_t \Vert_\infty + \mathbf{t}^i_3 \Vert \mathbf{M}_t \Vert_1 \mathbf{w}_\mathrm{max}
    %   + \mathbf{t}^i_0 \Vert \mathbf{w}_t \Vert_\infty \Bigg ) \leq f^x_i, \label{eq:state_robcon}\\
       & ~~~~~~~~~~~~~~~~ \eqref{N1_con},\eqref{eq:input_robcon}~\textnormal{if $N_t = 1$},\\
       &~~~~~~~~~~~~~~~~ \eqref{Ng1_con},\eqref{eq:input_robcon}~\textnormal{if $N_t > 1$},\\
	   % &~~~~~~~~~~~~ {x}_{t+N|t} \in {\mathcal{X}}_N,\\
	   % &~~~~~~~~~~~~ \mathbf{s}_t^j = [(s_t^j)^\top, (\hat{s}_t^j)^\top]^\top \geq 0,\\
	   % & ~~~~~~~~~~~~ \forall w_{k|t} \in {\mathbb{W}}, ~\forall \Delta_A \in \Phi_A,~\forall \Delta_B \in \Phi_B,\\
        % &~~~~~~~~~~~~~~~~~~ \forall k \in \{t,t+1,\dots,t+N_t-1\},  \\
        &~~~~~~~~~~~~~~~~~\bar{x}_{t|t} = x_t,
	\end{aligned}
\end{equation}
where we have denoted 
\begin{align*}
    & \mathbf{z}_t^{(N_t)} = \begin{bmatrix} (\bar{\mathbf{x}}^{(N_t)}_t)^\top &  \bar{x}_{t+N_t|t}^\top &  (\bar{\mathbf{u}}^{(N_t)}_t)^\top \end{bmatrix}^\top, \\
    & \mathbf{t}^{(N_t)} = \{\mathbf{t}^{(N_t)}_{w},  \mathbf{t}^{(N_t)}_{1},  \mathbf{t}^{(N_t)}_{2},  \mathbf{t}^{(N_t)}_{3}\}, \\
    & \bar Q^{(N_t)} = \text{diag}(I_{N_t} \otimes P, P_N, I_{N_t} \otimes R),\\
    & G^{(N_t)}_\mathrm{eq} = \begin{bmatrix} I_d & 0 & 0 & \cdots & 0 & 0 & 0 & 0 & \cdots & 0 \\ -\bar{A} & I_d & 0 & \cdots & 0 & 0 & -\bar{B} & 0 & \cdots & 0 \\ 0 & -\bar{A} & I_d & \cdots & 0 & 0 & 0 & -\bar{B} & \cdots & 0 \\ \vdots & \vdots & \vdots & \ddots & \vdots & \vdots & \vdots & \vdots & \ddots & \vdots \\ 0 & 0 & 0 & \cdots & -\bar{A} & I_d & 0 & 0 & \cdots & -\bar{B} \end{bmatrix}, \\
    & b^{(N_t)}_\mathrm{eq} = \begin{bmatrix} I_d \\ 0_{dN_t \times d} \end{bmatrix}.
\end{align*}
Note, we consider $\mathbf{t}^{(1)}=0$. We solve problem \eqref{eq:MPC_R_fin_trac} utilizing duality of convex programs \cite{ben2009robust}. This is detailed in \ref{duality} in the Appendix. The constraint tightenings in \eqref{eq:MPC_R_fin_trac} used in the robust state constraints are functions of the decision variables. 
% and can be optimally chosen by the solver while solving \eqref{eq:MPC_R_fin_trac}.
% , unlike Uncertainty Feedback Robust MPC (UF-Robust MPC) (see details in Section~\ref{sec:numerics}). 
This is the key contribution of our proposed approach.

We assume that \eqref{eq:MPC_R_fin_trac} is feasible at time step $t=0$ with $N_0 = N$. For $t \geq 1$, we apply the following policy
\begin{align}\label{eq:cl_control}
    u^\mathrm{MPC}_t(x_t) = \begin{cases} \bar{u}^\star_{t|t},~\textnormal{if \eqref{eq:MPC_R_fin_trac} is feasible,} \\ u^\star_{t|t_\mathrm{f}}(x_t),~\textnormal{otherwise} \end{cases}
\end{align}
to system \eqref{eq:unc_system}, where $t_\mathrm{f} \in \{0,1,\dots, N-1\}$ is the latest time step where \eqref{eq:MPC_R_fin_trac} was feasible previously. Thus, the time-shifted optimal policy from a previous time step is utilized as a safe backup, in case \eqref{eq:MPC_R_fin_trac} loses feasibility. As we cannot measure $w_t$ due to the presence of matrix uncertainties in \eqref{eq:unc_system}, see \cite[Section~5]{Goulart2006} for how to obtain the backup policy in state feedback form required for implementation. We then resolve \eqref{eq:MPC_R_fin_trac} at the next time step $(t+1)$ for horizon lengths $N_{t+1}$ obtained from \eqref{eq:n_t}. The control algorithm is summarized in Algorithm~\ref{alg2}.
%%%%%%%%%%%%%%%%%%%%%%%%%%%%%%%%%%%%%%%%%
\begin{algorithm}
\begin{algorithmic}
\State{\textbf{Inputs:}} $x_t, N, \mathbb{W}, \mathcal{X}_N, \mathbf{t}^{(N_t)}, \forall N_t \in \{2,3,\dots, N\}$
\State{\textbf{Initialize:}} $t=0$ 
 \While{$t \geq 0$}
      \State{Set horizon length $N_t$ from \eqref{eq:n_t};}
      \State{Solve MPC problem \eqref{eq:MPC_R_fin_trac};} 
      \State{Apply closed-loop input \eqref{eq:cl_control} to \eqref{eq:unc_system};}
      \State{Set $t = t+1$};
 \EndWhile
 \State{\textbf{end}}
\end{algorithmic}
\caption{Robust MPC for Linear Systems with Parametric and Additive Uncertainty}
\label{alg2}
\end{algorithm}
%%%%%%%%%%%%%%%%%%%%%%%%%%%%%%%%%%%%%%%%%%%%%%%%
\begin{remark}\label{rem:ver_enum}
Recall \eqref{err_in_sets_pol}--\eqref{eq:pol_out_termset_cond}. For time invariant $\Delta_A^\mathrm{tr}$ and $\Delta_B^\mathrm{tr}$ one may also efficiently enumerate all possible vertex sequences of $\mathbf{\Delta}_A$ and $\mathbf{\Delta}_B$ for robustifying the term $\mathbf{F}^x_i \bar{\mathbf{A}}_1 \mathbf{\Delta}_A  \bar{\mathbf{x}}^{(N_t)}_t + \mathbf{F}^x_i \bar{\mathbf{A}}_1 \mathbf{\Delta}_B \mathbf{\bar{u}}^{(N_t)}_t$ in \eqref{eq:fin_state_con} (with policy \eqref{eq:inputParam_DF_OL}). This partially replaces the bounds \eqref{eq:deladelb_bounds} to lower conservatism. As we use the backup policy in \eqref{eq:cl_control} without requiring recursive feasibility of~\eqref{eq:MPC_R_fin_trac}, the number of such sequences is limited to the number of vertices characterizing the uncertain matrices (i.e., each vertex repeated $N_t$ times along the horizon), and is not combinatorial. See~\cite[Figure 3]{scokaert1998min} for further insights into why combinatorial enumerations are required otherwise. 
\end{remark}
%%%%%%%%
\section{Robust Constraint Satisfaction and Stability}\label{sec:feas_n_stab}
%%%%%%%%%%%%%%%%%%%%%%%%%%%%%%%%%
We first prove the robust satisfaction of constraints \eqref{eq:FTOCP_constr} for the closed-loop system \eqref{eq:unc_system} and \eqref{eq:cl_control}. Afterwards, we show the stability properties of the proposed robust MPC in Algorithm~\ref{alg2}. 
% For that purpose we utilize the following set of definitions and  assumptions. 
%%%%%%%%%%%%%%%%%%%%%%%%%%%%%%%%
\subsection{Feasibility of Robust Constraints}
\begin{theorem}\label{thm1}
Let optimization problem \eqref{eq:MPC_R_fin_trac} with tightened constraints \eqref{tight1} be feasible at time step $t=0$ for $N_t = N$, where the bounds $\{\mathbf{t}^{(N_t)}_{w}, \mathbf{t}^{(N_t)}_{1}, \mathbf{t}^{(N_t)}_{2}, \mathbf{t}^{(N_t)}_{3} \}$ are obtained by solving \eqref{eq:bound_mainterm}-\eqref{fourthbound}. Then, the closed-loop system~\eqref{eq:unc_system} and~\eqref{eq:cl_control} robustly satisfies state and input constraints~\eqref{eq:FTOCP_constr}, for all $t\geq 0$.
\end{theorem}
\begin{proof}
See \ref{ProofThe1} in the Appendix.  
\end{proof}
%%%%%%%%%%%%%%%%%%%%%%%%%%%%%%%%%%%%%
\subsection{Stability}\label{roa_section}
To prove stability of the origin for system~\eqref{eq:unc_system} in closed-loop with the MPC control law~\eqref{eq:cl_control}, we first introduce the following set of assumptions and definitions.
\begin{assumption}\label{assump:orig_in}
Denote the state and input constraints in \eqref{eq:FTOCP_constr} as $\mathcal{X}$ and $\mathcal{U}$. We assume the convex-compact sets $\mathcal{X}, \mathcal{U}$ and $\mathbb{W}$ contain the origin in their interior. 
\end{assumption}
%%%%%%%%%%%%%%%%%%%%%%%%%%%%
\begin{definition}[Robust Precursor Set]\label{def:preset}
Given a control policy $\pi(\cdot)$ and the closed-loop system $x_{t+1} = A x_t + B\pi(x_t) + w_t$ with $w_t \in \mathbb{W}$ for all $t\geq 0$, we denote the robust precursor set to the set $\mathcal{S}$ under a policy $\pi(\cdot)$ as
\begin{align}\label{eq:pre_set_eq}
\mathrm{Pre}(\mathcal{S},A,B,\mathbb{W}, \pi(\cdot)) = \{x \in \mathbb{R}^d: Ax + B\pi(x) + w \in \mathcal{S}, \forall w \in \mathbb{W}\}.
\end{align}
$\mathrm{Pre}(\mathcal{S},A,B,\mathbb{W}, \pi(\cdot))$ defines the set of states of the system $x_{t+1} = A x_t + B\pi(x_t) + w_t$, which evolve into the target set $\mathcal{S}$ in one time step for all $w_t \in \mathbb{W}$.
\end{definition}
\begin{definition}[$N$-Step Robust Controllable Set]\label{def:RobustPre} Given a control policy $\pi(\cdot)$ and the closed-loop system $x_{t+1} = A x_t + B\pi(x_t) + w_t$, we recursively define the $N$-Step Robust Controllable set to the set $\mathcal{S}$ as 
\begin{equation*}
\begin{aligned}
    & \mathcal{C}_{t\rightarrow t+k+1}(\mathcal{S}) = \mathrm{Pre}(\mathcal{C}_{t\rightarrow t+k}(\mathcal{S}), A, B, \mathbb{W}, \pi(\cdot)) \cap \mathcal{X},\\[1ex]
    & \textnormal{with } \mathcal{C}_{t\rightarrow t}(\mathcal{S})=\mathcal{S},
\end{aligned}
\end{equation*}
for $k \in \{0, 1, \dots, N-1 \}$. 
\end{definition}
\noindent The $N$-Step Robust Controllable set $\mathcal{C}_{t\rightarrow t+N}(\mathcal{S})$ collects the states satisfying the state constraints which can be steered to the set $\mathcal{S}$ in $N$ steps under the policy $\pi(\cdot)$. 
%%%%%%%%%%%%%%%%%%%%%%%%%%%%%%%%%%%%%
% \begin{remark}\label{lastrem}
% The $N_1$-Step Robust Controllable Set to $\mathcal{X}_N$ under the MPC policy
% \begin{align}\label{pol_remark}
%     \pi(x_t) = u^\star_{t}(x_t) = \bar{u}^\star_{t|t},~\forall t \geq 0,
% \end{align}
% synthesized by solving \eqref{eq:MPC_R_fin_trac} with tightened constraints \eqref{tight1} for a fixed horizon length $N_1$,
% % \eqref{eq:cl_control} 
% is \emph{not} necessarily a subset of the corresponding $N_2$-Step Robust Controllable Set, for any $N_2 > N_1$. In other words, any $N$-Step Robust Controllable set to $\mathcal{X}_N$ under MPC policy \eqref{pol_remark} is \emph{not} a robust control invariant set \cite[Chapter~10]{borrelli2017predictive}. This is different from standard robust MPC methods  \cite{chisci2001systems,Goulart2006,langson2004robust}  because the description of the uncertainty along the horizon (see \eqref{tdel1_etc} and bounds \eqref{eq:bound_mainterm}-\eqref{fourthbound} in the Appendix) is different and more conservative compared to the one used in the calculation of the terminal robust positive invariant set in \eqref{eq:term_set_DF}.
% % This  motivates the use of an adaptive horizon approach in \eqref{eq:MPC_R_fin_trac} in  Section~\ref{ssec:mpc_problem}. 
% \end{remark}
%%%%%%%%%%%%%%%%%%%%%%%%%%%%%%%%%%%%%
\begin{assumption}\label{assump:stagecost} 
The matrices $P$ and $R$ defining the stage cost $\ell(x, u) = x^\top P x + u^\top R u$ satisfy $P\succ0$, $R\succ0$.
% The stage cost $\ell(\cdot, \cdot)$ in \eqref{eq:MPC_R_fin_trac} is chosen as $\ell(x, u) = x^\top P x + u^\top R u$ for some $P=P^\top \succ 0$ and $R=R^\top \succ 0$, which is continuous and positive definite in domain $\mathcal{R} \times \mathcal{U}$. 
\end{assumption}
%%%%%%%%%%%%%%%%%%%%%%%%%%%%%%%%%%
\begin{assumption}\label{assump: termcost}
The matrix $P_N$ which defines the terminal cost in \eqref{eq:MPC_R_fin_trac} is chosen as
% \begin{align}
%     Q(x) = x^\top P_N x, 
% \end{align}
% where 
a matrix $P_N \succ 0$ satisfies 
\begin{align}
 & x^\top \Big (-P_N + (P+K^\top R K) + \bar{A}_\mathrm{cl}^\top P_N \bar{A}_\mathrm{cl} \Big ) x \leq 0,~\forall x \in \mathcal{X}_N \nonumber,
\end{align}
where $\bar{A}_\mathrm{cl} = \bar{A} + \bar{B}K$.
% The condition \eqref{eq:lmi_ly} can be satisfied by solving a Linear Matrix Inequality \cite{boyd1994linear}.
\end{assumption}
%%%%%%%%%%%%%%%%%%%%%%%%%%%%%%%%%%%
\begin{definition}[Input to State Stability (ISS) \cite{lin1995various}]
Consider \eqref{eq:unc_system} in closed-loop with the MPC law \eqref{eq:cl_control}, obtained from \eqref{eq:MPC_R_fin_trac} with tightened constraints \eqref{tight1}:
\begin{align}\label{eq:cl_loop_system}
    x_{t+1} = Ax_t + B{u}^\mathrm{MPC}_t(x_t) + w_t,~\forall t\geq 0.
\end{align}
We say that the origin of the closed-loop system \eqref{eq:cl_loop_system} is ISS in $\mathcal{X}_N$ if for all $\Vert \tilde{w}_t\Vert_\infty \leq \tilde{w}_\mathrm{max}$, $t \geq 0$, $x_0 \in \mathcal{X}_N$
\begin{equation*}
    \Vert x_{t+1}\Vert \leq \beta(\Vert x_0\Vert, t+1) + \gamma \big( \Vert \tilde{w}_i\Vert_{\mathcal{L}_\infty} \big),
\end{equation*}
where $\tilde{w}_i = \Delta^\mathrm{tr}_A x_i + \Delta^\mathrm{tr}_B u_i + w_i$, $\Vert \tilde{w}_i \Vert_{\mathcal{L}_\infty} = \sup_{i \in \{0,\dots,t\}}\Vert \tilde{w}_i \Vert$, and  $\beta(\cdot, \cdot)$ and $\gamma(\cdot)$ are class-$\mathcal{KL}$ and class-$\mathcal{K}$ functions.
\end{definition}
%%%%%%%%%%%%%%%%%%%%%%%%%%%%%%%
% Note that 
% % in standard MPC strategies with quadratic cost, the finite time optimal control problem can be reformulated as a parametric Quadratic Program (QP). This fact is used in~\cite{Goulart2006} to show continuity of the value function and then to prove ISS of the origin. 
% in the proposed approach, the value function from Algorithm~\ref{alg2}, i.e., the optimal cost $J^{\star}(x_t, \mathbf{t}^{(N^\star_t)})$ is \emph{not} given by the solution to a parametric QP. Therefore its continuity cannot be guaranteed, and the standard technique from~\cite{Goulart2006} cannot be used to prove ISS of the origin. Instead, we use the following modified definition of an ISS Lyapunov function, which requires continuity of the value function only at the origin.
% \vspace{3pt}
%%%%%%%%%%%%%%%%%%%%%%%%%%%%%%%%%%%
\begin{definition}[ISS Lyapunov Function \cite{lin1995various}]\label{iss_lyap}
Consider the closed-loop system in \eqref{eq:cl_loop_system}. Then the origin is ISS in $\mathcal{X}_N$, if there exists class-$\mathcal{K}_\infty$ functions $\alpha_1(\cdot)$, $\alpha_2(\cdot)$, $\alpha_3(\cdot)$, a class-$\mathcal{K}$ function $\sigma(\cdot)$ and a function $V(\cdot): \mathbb{R}^d \mapsto \mathbb{R}_{\geq 0}$ continuous in $\mathcal{X}_N$, such that, 
\begin{align*}
    & \alpha_1(\Vert x \Vert ) \leq V(x) \leq \alpha_2(\Vert x \Vert ),~\forall x \in \mathcal{X}_N,\\
    & V(x_{t+1}) - V(x_t) \leq -\alpha_3(\Vert x_t \Vert) + \sigma(\Vert \tilde{w}_i \Vert_{\mathcal{L}_\infty}).
\end{align*}
Function $V(\cdot)$ is called an ISS Lyapunov function for the closed-loop system \eqref{eq:cl_loop_system}.
\end{definition}
%%%%%%%%%%%%%%%%%%
% \textcolor{blue}{In the following theorem we prove that the origin is ISS for the closed-loop system \eqref{eq:cl_loop_system}.}

% \textcolor{blue}{
% \begin{lemma}\label{lem1}
% Consider the closed-loop system \eqref{eq:cl_loop_system}. Then, $x_t \in \mathcal{X}_N$ for $t=N$.
% \end{lemma}
% \begin{proof}
% Write. 
% \end{proof}
% }
%%%%%%%%%%%%%%%%%%
% \textcolor{blue}{
% \begin{theorem}\label{isstheorem}
% Let Assumptions~\ref{assump:stable}-\ref{assump: termcost} hold and let $x_0 \in \mathcal{R}$. From Lemma~\ref{lem1} we can guarantee $x_{t} \in \mathcal{X}_N$ at $t=N$. Then for all $t \geq N+1$, the optimal cost of \eqref{eq:MPC_R_fin_trac} with $N_t = 1$, i.e., $V^{\mathrm{MPC}}_{t \rightarrow t+N_t}(x_t, \mathbf{t}^1, 1)$ is an ISS Lyapunov function for closed-loop system \eqref{eq:cl_loop_system}. This guarantees Input to State Stability of the origin. 
% \end{theorem}
% }
\begin{theorem}\label{isstheorem}
Let Assumptions~\ref{assump:stable}-\ref{assump: termcost} hold and let the optimization problem \eqref{eq:MPC_R_fin_trac} be feasible at time step $t=0$ with $N_t = N$. Then, $x_t \in \mathcal{X}_N$ for all $t \geq N$ and the origin of the closed-loop system~\eqref{eq:cl_loop_system} is ISS.
\end{theorem}

\begin{proof}
See \ref{ProofISS} in the Appendix.
\end{proof}

\section{The ROA and Its Inner Approximation} \label{sec:roa_sec}
We define the Region of Attraction (ROA) for Algorithm~\ref{alg2}, denoted by $\mathcal{R}$, as the $N$-Step Robust Controllable Set to the terminal set $\mathcal{X}_N$ under the policy \eqref{eq:cl_control} for $t = 0$. This ensures that from Theorem~\ref{thm1} and Theorem~\ref{isstheorem} we have $\forall w_t \in \mathbb{W}$:
\begin{align*}
    x_0 \in \mathcal{R} \implies & \begin{cases} x_t \in \mathcal{X},~\forall t \geq 0,~\textnormal{and}\\
     x_t \in \mathcal{X}_N \subseteq \mathcal{X},~\forall t \geq N, \end{cases}
\end{align*}
where $x_{t+1} = Ax_t + B{u}^\mathrm{MPC}_t(x_t) + w_t$ for all $t \geq 0$. Thus, all the initial states in the ROA are steered to the terminal set $\mathcal{X}_N$ in maximum of $N$-steps while robustly satisfying \eqref{eq:FTOCP_constr}, where the origin of \eqref{eq:cl_loop_system} is ISS. 
%%%%%%%%%%%%%%%%%%%%%%%%%%%%%%%%%%%%%%%%
The ROA can be computed by solving problem~\eqref{eq:MPC_R_fin_trac} as a parametric optimization problem, with parameter $x_t$~\cite{borrelli2017predictive}. However, this computation may be prohibitive. We therefore use the fact that the ROA is convex and obtain its inner approximation using a set of vectors, following \cite{rosolia2019robust}. Along each vector, we find an initial state for which \eqref{eq:MPC_R_fin_trac} is feasible and which minimizes the inner product with the vector. The ROA is then approximated as the convex hull of these states. This is elaborated below. 

Given a vector $v\in \mathbb{R}^d$, we define the following optimization problem at time step $t=0$:
\begin{equation}\label{eq:regAttApp}
\begin{aligned}
& P(N, v) = \\
& \min_{\substack{x_0, \mathbf{M}^{(N)}_0, \bar{\mathbf{u}}^{(N)}_0\\ \bar{\mathbf{x}}^{N}_0 }} ~~ v^\top  x_0  \\
& ~~~~~~\textrm{s.t.,} ~~ (v^\perp )^\top x_0 =0, \\
& ~~~~~~~~~~~~~~ G^{(N)}_\mathrm{eq}\begin{bmatrix} (\bar{\mathbf{x}}^{(N)}_0)^\top & \bar{x}^\top_{N|0} & (\bar{\mathbf{u}}^{(N)}_0)^\top\end{bmatrix}^\top = b^{(N)}_\mathrm{eq} x_0, \\
& ~~~~~~~~~~~~~~ \bar x_{0|0}=x_0,\\
% & ~~~~~~~~~~~~~~\bar{x}_{k+1|0} = \bar{A}\bar{x}_{k|0} + \bar{B}\bar{u}_{k|0},  \\
% & ~~~~~~~~~~~~~~u_{k|0} (x_{k|0})= \sum \limits_{l=0}^{k-1}M_{k,l|0} w_{l|0}  + \bar{u}_{k|0},  \\
%& ~~~~~~~~~~~~~~~ \max_{\mathbf{w}_t \in \mathbf{W}} \Bigg ( F^x_i(\bar{\mathbf{A}} \bar{\mathbf{x}}_t + \bar{\mathbf{B}} (\mathbf{M}_t \mathbf{w}_t + \bar{\mathbf{u}}_t) + \mathbf{w}_t) + (\mathbf{t}^i_{\delta A} + \mathbf{t}^i_1) \Vert \bar{\mathbf{x}}_t \Vert_\infty + ( \mathbf{t}^i_{\delta B} + \mathbf{t}^i_2) \Vert \mathbf{M}_t \Vert_1 \mathbf{w}_\mathrm{max} + \nonumber \\
%& ~~~~~~~~~~~~~~~~~~~~~~~~~~~~~~~~ \cdots + (\mathbf{t}^i_{\delta B} + \mathbf{t}^i_2) \Vert \bar{\mathbf{u}}_t \Vert_\infty + \mathbf{t}^i_3 \Vert \mathbf{M}_t \Vert_1 \mathbf{w}_\mathrm{max}
    %   + \mathbf{t}^i_0 \Vert \mathbf{w}_t \Vert_\infty \Bigg ) \leq f^x_i, \label{eq:state_robcon}\\
% & ~~~~~~~~~~~~~\max_{\substack{\mathbf{w}_0 \in \mathbf{W}\\ \Delta_A \in \mathcal{P}_A \\ \Delta_B \in \mathcal{P}_B}}  F^x((\bar{\mathbf{A}}+ \bar{\mathbf{A}}_1 \mathbf{\Delta}_A) \bar{\mathbf{x}}^{(N)}_0 + (\bar{\mathbf{B}} + \bar{\mathbf{A}}_1 \mathbf{\Delta}_B) (\mathbf{M}^{(N)}_0 \mathbf{w}_0 + \bar{\mathbf{u}}^{(N)}_0) + (\bar{\mathbf{A}}_1 - \mathbf{I}_d)\bar{\mathbf{B}}\mathbf{M}^{(N)}_0 \mathbf{w}_0 + \mathbf{w}_0) \leq f_\mathrm{tight}^x, \\
% &~~~~~~~~~~~~~\max_{\mathbf{w}_0 \in \mathbf{W}} \mathbf{H}^u \Big ( \mathbf{M}^{(N)}_0 \mathbf{w}_0 + \bar{\mathbf{u}}^{(N)}_0 \Big) \leq \mathbf{h}^u,\\
& ~~~~~~~~~~~~~~ \eqref{eq:state_robcon},\eqref{eq:input_robcon},~\textnormal{(with $N_0 = N$)},\\
% & ~~~~~~~~~~~~~~\forall k \in \{0,1, \dots, N-1\},
\end{aligned}
\end{equation}
with $\mathbf{f}^x_{\mathrm{tight}}$ chosen as per \eqref{tight1}, where $v^\perp \in \mathbb{R}^d$ is a vector perpendicular to $v \in \mathbb{R}^d$.  Therefore, given a user-defined set of vectors $\mathcal{V}=\{v^{(1)},v^{(2)},\dots,v^{(n)}\}$, problem~\eqref{eq:regAttApp} can be solved repeatedly and the convex hull of the optimal initial states $x^\star_0$ is an inner approximation to the ROA. 
% Algorithm~\ref{RegAttrAppr} summarizes this procedure.
\begin{algorithm}[h]
\begin{algorithmic}
% 	\SetAlgoLined
	\State{\textbf{Inputs:}} Vectors $\mathcal{V}=\{v^{(1)}, v^{(2)}, \dots, v^{(n)}\}$ and $N$
	\State{\textbf{Initialize:}} $\mathcal{R}_\mathrm{ap} = \varnothing$ 
	\For{$v^{(i)} \in \mathcal{V}$}
	
  	\State{Solve $P(N, v^{(i)})$ from~\eqref{eq:regAttApp}. Let $x_0^\star$ be the optimal initial state from $P(N, v^{(i)})$.
	
	 \State{Set $\mathcal{R}_\mathrm{ap} = \text{conv}\{\mathcal{R}_\mathrm{ap} \cup \{x_0^\star\}\}$}.}

	\EndFor 
	\State \textbf{end}

	\State\textbf{Output:} Approximate ROA: $\mathcal{R}_\mathrm{ap} \subseteq \mathcal{R}$.

\end{algorithmic}
\caption{Approximate ROA}
\label{RegAttrAppr}
\end{algorithm}
It is clear from Algorithm~\ref{RegAttrAppr} that the ROA approximation can improve, as the number of vectors in $\mathcal{V}$ increases. 
%%%%%%%%%%%%%%%%%%%%%%%%%%%%%%%%%%%%%%%%
\section{Numerical Simulations}\label{sec:numerics}
We present our numerical simulations in this section (\href{https://github.com/monimoyb/RMPCPy}{Link to GitHub Repository}). Algorithm~\ref{alg2} is implemented with $N= 3$ and $N_t$ chosen as per \eqref{eq:n_t} for all $t \geq 0$. We compare the performance of our Algorithm~\ref{alg2} with that of the finite dimensional constrained LQR algorithm of \cite[Section~2.3]{dean2018safely}, and also with a tube MPC of \cite[Section~5]{langson2004robust}.
For our comparisons, we compute approximate MPC solutions to the problem:
\begin{equation}\label{eq:generalized_InfOCP_ex}
	\begin{array}{llll}
		\hspace{0cm}    
			\hspace{0cm} 
	\displaystyle\min_{u_0,u_1(\cdot),\ldots} & \displaystyle\sum\limits_{t\geq 0} 10\left \| \bar{x}_t \right\|^2_2 + 2 \left\| u_t(\bar{x}_t) \right\|^2_2   \\[1ex]
	\hspace{5mm}	\text{s.t.,} & x_{t+1} = Ax_t + Bu_t(x_t) + w_t, \\
	&\textnormal{with } A = \bar{A} + \Delta_A,~B = \bar{B} + \Delta_B, \\
	& \bar{x}_{t+1} = \bar{A}\bar{x}_t + \bar{B}u_t(\bar{x}_t),\\
	& \begin{bmatrix} -8 \\ -8 \\ -4
	\end{bmatrix} \leq \begin{bmatrix}x_t \\ u_t(x_t)
	\end{bmatrix} \leq \begin{bmatrix}8 \\ 8 \\ 4
	\end{bmatrix},\\[3.5ex]
	& \forall w_t \in \mathbb{W},~\forall \Delta_A \in \mathcal{P}_A,~\forall \Delta_B \in \mathcal{P}_B, \\
	& x_0 = x_S,~t=0,1,\ldots,
	\end{array}
\end{equation}
with disturbance set $\mathbb{W} = \{w: \Vert w \Vert_\infty \leq 0.1\}$, where 
\begin{align*}
\bar{A} = \begin{bmatrix} 1 & 0.15\\
0.1 & 1\end{bmatrix},~\bar{B} = \begin{bmatrix} 0.1\\ 1.1 \end{bmatrix},~{A} = \begin{bmatrix} 1 & 0.05\\
0 & 1 \end{bmatrix},~{B} = \begin{bmatrix} 0\\ 1.1 \end{bmatrix}.
\end{align*}
For solving \eqref{eq:generalized_InfOCP_ex}we consider the uncertainty sets
\begin{align*}
    & \mathcal{P}_A = \mathrm{conv} \Big (\begin{bmatrix} 0& \pm 0.1 \\ \pm 0.1 & 0 \end{bmatrix} \Big ),~\textnormal{(4 matrices)}\\
    & \mathcal{P}_B = \mathrm{conv} \Big (\begin{bmatrix} 0 \\ \pm 0.1 \end{bmatrix}, \begin{bmatrix} \pm 0.1 \\ 0 \end{bmatrix} \Big )~\textnormal{(4 matrices)} \nonumber.
\end{align*}
That is, we consider uncertainty in only the off-diagonal terms of $\bar{A}$, assuming that the diagonal terms are known. The equivalent uncertainty sets $\Delta_A \in \Phi_{A, \infty}$ and $\Delta_B \in \Phi_{B, \infty}$ considered in \cite{dean2018safely} are given by 
\begin{align*}
    & \Phi_{A,\infty} = \{\phi \in \mathbb{R}^{2 \times 2}: \max_{x \neq 0} \frac{\Vert\phi x\Vert_\infty}{\Vert x\Vert_\infty} \leq 0.1\},\\
    &\Phi_{B,\infty} = \{\phi \in \mathbb{R}^{2 \times 1}: \max_{x \neq 0} \frac{\Vert\phi x\Vert_\infty}{\Vert x\Vert_\infty} \leq 0.1\}.
\end{align*}
For this example, we utilize Remark~\ref{rem:ver_enum}. Gain $K$ for constructing the terminal set $\mathcal{X}_N$ is chosen as $K = -[0.452, 0.418]$. The fixed point iteration algorithm computing $\mathcal{X}_N$ converges in 7 iterations.

\subsection{Comparison with Tube MPC \cite{langson2004robust}}
For this comparison, we choose a horizon of 5 for the tube MPC method in \cite[Section~5]{langson2004robust}. The tube cross section parameter $Z$ is chosen as the minimal robust positive invariant set \cite[Definition~3.4]{kouvaritakis2016model} for system \eqref{eq:unc_system} under a feedback $u = -[1.2604, 0.7036]x$, and the terminal set $\mathcal{X}_f$ is chosen as our terminal set $\mathcal{X}_N$ constructed with \eqref{eq:term_set_DF}. See \cite{langson2004robust} for details on these quantities.
% This ensures that the feasibility guarantees of \cite[Proposition~5]{langson2004robust} hold. 
%%
\begin{figure}[h!]
\centering\textcolor{yellow}{\textbf{$\blacksquare$}} Algorithm~\ref{alg2} ~ \textcolor{gray}{\textbf{$\blacksquare$}} Tube MPC ~ \textcolor{blue}{\textbf{$\blacksquare$}} Constrained LQR \\[0.2cm]
	\centering
	\includegraphics[width=12cm]{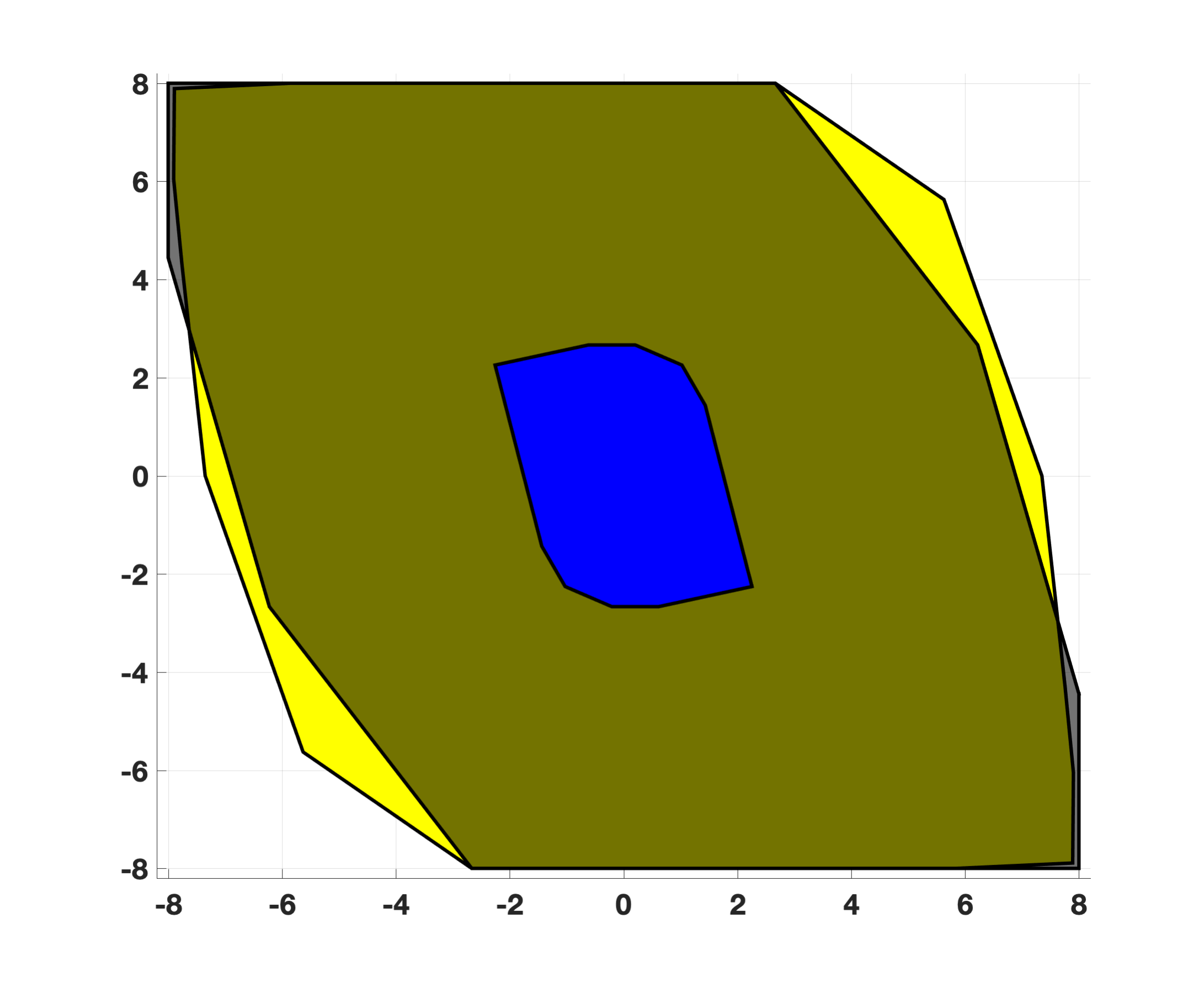}%\\
	%\centering\textcolor{blue}{\textbf{$\blacksquare$}} Approx. ROA of Algorithm 2 ~ \textcolor{red}{\textbf{$\blacksquare$}} Approx. ROA of Controller in \cite{dean2018safely}
	\caption{Comparison of the Approximate Region of Attraction of Algorithm~\ref{alg2} and the convex hull of the feasible initial state samples with tube MPC in \cite[Section~5]{langson2004robust} and constrained LQR in \cite[Section~2.3]{dean2018safely}. 
% 	The Approximate Region of Attraction of Algorithm~\ref{alg2} is obtained as per \eqref{ss_in} with $N_0 \in \{1,2,3,4,5\}$. 
	}
	\label{fig:init_sample_slsTube}    
\end{figure}
Recall the notion of the ROA of Algorithm~\ref{alg2} from Section~\ref{sec:roa_sec} and also its inner approximation obtained from Algorithm~\ref{RegAttrAppr}. We now choose a set of $N_\mathrm{init} = 100$ initial states $x_S$, created by a $10 \times 10$ uniformly spaced grid of the set of state constraints in \eqref{eq:generalized_InfOCP_ex}. From  each of these initial state samples we check the feasibility of the tube MPC control problem in \cite[Section~5]{langson2004robust}. The source code to solve the tube MPC is adapted from \cite{slscode}. 
% %%
% \begin{figure}[h!]
% \centering\textcolor{yellow}{\textbf{$\blacksquare$}} Algorithm~\ref{alg2} ~ \textcolor{gray}{\textbf{$\blacksquare$}} Tube MPC ~ \textcolor{blue}{\textbf{$\blacksquare$}} Constrained LQR \\[0.2cm]
% 	\centering
% 	\includegraphics[width=\columnwidth]{ol_ss2STube.eps}%\\
% 	%\centering\textcolor{blue}{\textbf{$\blacksquare$}} Approx. ROA of Algorithm 2 ~ \textcolor{red}{\textbf{$\blacksquare$}} Approx. ROA of Controller in \cite{dean2018safely}
% 	\caption{Comparison of the Approximate Region of Attraction of Algorithm~\ref{alg2} and the convex hull of the feasible initial state samples with tube MPC in \cite[Section~5]{langson2004robust} and Constrained LQR in \cite[Section~2.3]{dean2018safely}. 
% % 	The Approximate Region of Attraction of Algorithm~\ref{alg2} is obtained as per \eqref{ss_in} with $N_0 \in \{1,2,3,4,5\}$. 
% 	}
% 	\label{fig:init_sample_slsTube}    
% \end{figure}
% %%
The convex hull of the feasible initial state samples, which inner approximates its ROA, is then compared to the approximate ROA of Algorithm~\ref{alg2}. This comparison is shown in Fig.~\ref{fig:init_sample_slsTube}. The approximate ROA of Algorithm~\ref{alg2} is about 1.04x in volume of that of the tube MPC in \cite[Section~5]{langson2004robust}. The advantage of our approach becomes clearer in Table~\ref{tab:compTime}.  
\begin{table}[h!]
\caption{Average computation times [sec] comparison. Values are obtained with a MacBook Pro 16inch, 2019, 2.3 GHz 8-Core Intel Core i9, 16 GB memory, using the Gurobi solver \cite{gurobi2015gurobi}.}
\label{tab:compTime}
  \begin{center} 
  \begin{tabular}{|c|c|c|c|}
    \hline
   \multirow{2}{4em}{\textbf{Horizon}} & \multicolumn{2}{c|}{\textbf{Algorithm~\ref{alg2}}} & {\textbf{Tube MPC in \cite{langson2004robust}}}\\
    % \hline
    % \textbf{Inactive Modes} & \textbf{Description}\\
    \cline{2-4}
    & \textbf{online} & \textbf{offline} & \textbf{online}\\
    %\hhline{~--}
    \hline
     $N_t=1$   &  0.0019  & 0 &  0.0054\\
     \hline
     $N_t=2$   &  0.0058  & 0.0279 &  0.1042 \\ \hline
     $N_t=3$   &  0.0111  & 0.0687 & 0.2057 \\ \hline
    %  $N_t=4$   &  0.0066  & 2.5916 &  0.2181 & 0 \\ \hline
    %  $N_t=5$   &  0.0084  & 803.28 & 0.2814 & 0\\ \hline
  \end{tabular}
  \end{center}
\end{table}

We see from Table~\ref{tab:compTime} that for all relevant horizon lengths $N_t \in \{1,2,3\}$, solving \eqref{eq:MPC_R_fin_trac} is cheaper than computing the tube MPC online, even after adding the offline computation times required for bounds \eqref{eq:bound_mainterm}-\eqref{fourthbound}.
%%
% This indicates that the proposed Algorithm~\ref{alg2} is able to
% % significantly improve over the robust MPC method based on \cite{Goulart2006}, and 
% attain a set of feasible initial conditions larger than the tube MPC in \cite[Section~5]{langson2004robust}. 
%%%%%%%%%%%%%%%%%%%%%%%%%%%%%%%%%%%%%%%%%%%%

\begin{remark}
Tube MPC methods such as \cite{langson2004robust, munoz2013recursively, rakovic2013homothetic, fleming2014robust} also require offline matrix/set computations before online control design. See  \cite[Chapter~5]{kouvaritakis2016model} for further details. In the considered example, computing the set $Z$ for the tube MPC in \cite{langson2004robust} required about 49 seconds offline. However, we have chosen not to include this in the comparison in Table~\ref{tab:compTime}, as any alternative simpler choice of $Z$ is also valid. The choice of $Z$ affects the ROA \cite{slsmpc}.   
\end{remark}

\subsection{Comparison with Constrained LQR \cite{dean2018safely}}
Using the same 100 initial state samples, we now check the feasibility of the constrained LQR synthesis problem in \cite[Section~2.3]{dean2018safely}. 
%%%%%%%%%%%%%%%%%%%%%%%%%%%%%%%%%%%%%%%%%%%%%%%%%
We run all the simulations for an FIR length (same as control horizon length) of $L=15$. The values of parameters for constraint tightenings are chosen as $\tau = 0.99$ and $\tau_\infty = 0.2$ after a grid search. See \cite[Problem~2.8]{dean2018safely} for further details on these parameters. 
%%%%%%%%%%%%%%%%%%%%%%%%%%%%%%%%%%%%%%%%%%%%%%%
The convex hull of the feasible initial state samples with the algorithm of \cite[Section~2.3]{dean2018safely}, which inner approximates its ROA, is about $12$ times smaller in volume and is a \emph{subset} of the approximate ROA of Algorithm~\ref{alg2}, as seen in Fig.~\ref{fig:init_sample_slsTube}. Furthermore, as \cite{dean2018safely} does not solve any optimization problem for control synthesis for $t\geq 1$, we highlight that this gain in ROA volume can also be obtained with an open-loop policy given by: 
\begin{align}\label{sarahpol_backup}
  \Pi^\mathrm{safe}_\mathrm{ol}(x_t) = \begin{cases} u^\star_{t|0}(x_t), &\mbox{if } t \leq ({N}-1), \\
    Kx_t, & \textnormal{otherwise}. \end{cases}
\end{align}
System \eqref{eq:unc_system} with \eqref{sarahpol_backup} maintains robust satisfaction of \eqref{eq:FTOCP_constr} for all time steps, without re-solving \eqref{eq:MPC_R_fin_trac} after $t=0$. Moreover, the one-time control computation times with \cite{dean2018safely} are comparable. We omit showing these values due to the difference in programming languages.

%%%%%%%%

% The offline computational expenses in Table~\ref{tab:compTime} can be lowered further with bounds \eqref{eq:bound_maintermD}-\eqref{fourthboundD} at the cost of higher conservatism in \eqref{eq:MPC_R_fin_trac}, as discussed in~\ref{bin_bound} in the Appendix. 
%%
\section{Conclusions}
We proposed a novel approach to design a robust Model Predictive Controller (MPC) for constrained uncertain linear systems. The uncertainty considered included both mismatch in the system dynamics matrices, and additive disturbance. 
% With set based bounds for each component of the model uncertainty being known at the time of control design, 
% We proposed a novel optimization based constraint tightening strategy utilizing these bounds. 
The proposed MPC achieved robust satisfaction of the imposed state and input constraints for all realizations of the model uncertainty.
% , while avoiding 
% % planning over nominal trajectories that avoid 
% restrictive constraint tightenings. 
We further proved Input to State Stability of the origin. With numerical simulations, we demonstrated that our controller obtained at least 3x and up to 20x speedup in online control computations and an approximately $4\%$ larger ROA by volume, compared to the tube MPC in \cite{langson2004robust}. We also demonstrated an approximately 12x decrease in conservatism over the constrained LQR algorithm of \cite{dean2018safely} using a safe open-loop policy.
% , 
%%%%%%%%%%%%%%%%%%%%%%%%%
\section*{Acknowledgements}
We thank Sarah Dean for the source code of the constrained LQR method. This project has received funding from the European Union’s Horizon 2020 research and innovation programme under the Marie Sk\l{}odowska-Curie grant agreement No 846421. This work was also funded by ONR-N00014-18-1-2833, and NSF-1931853. 
%%%%%%%%%%%%%%%%%%%%%%%%%%%%%%
% \renewcommand{\baselinestretch}{0.8}
\bibliographystyle{plain}
\bibliography{bibliography.bib}
%%%%%%%%%%%%%%%%%%%%%%%%%%%%%%
\appendix
\section{Appendix}
\subsection{Matrix Definitions}\label{A1}
The prediction dynamics matrices $\mathbf{A}^x, \mathbf{A}^u, \mathbf{A}^{\Delta u}$ and $\mathbf{A}^w$ in \eqref{eq:state_propagation} for a horizon length\footnote{Equation \eqref{eq:state_propagation} was introduced with a fixed horizon length of $N$, i.e., $\bar{N} \leftarrow N$. However, dimensions of these matrices vary as horizon length is varied later in Section~\ref{ssec:mpc_problem}.} of $\bar{N}$ are given by
\begin{align*}
    \mathbf{A}^x & = \begin{bmatrix}
    A_\Delta & 0 & 0 & \dots & 0\\
    A_\Delta\Delta_A & A_\Delta & 0 & \dots & 0\\
    A_\Delta^2\Delta_A & A_\Delta\Delta_A & A_\Delta & \dots & 0\\
    \vdots & \vdots & \vdots
    & \ddots & \vdots\\
    A_\Delta^{\bar{N}-1}\Delta_A & A_\Delta^{\bar{N}-2}\Delta_A & \dots & \dots & A_\Delta
    \end{bmatrix} \in \mathbb{R}^{d\bar{N} \times d\bar{N}},\\
    \mathbf{A}^u & = \begin{bmatrix}
    B_\Delta & 0 & 0 & \dots & 0\\
    A_\Delta \Delta_B & B_\Delta & 0 & \dots & 0 \\
    A_\Delta^2\Delta_B & A_\Delta\Delta_B & B_\Delta & \dots & 0 \\
    \vdots & \vdots & \vdots & \ddots & \vdots\\
    A_\Delta^{\bar{N}-1}
    \Delta_B & A_\Delta^{\bar{N}-2}
    \Delta_B & \dots & \dots & B_\Delta
    \end{bmatrix} \in \mathbb{R}^{d\bar{N} \times m\bar{N}},\\
    \mathbf{A}^{\Delta u} & = \begin{bmatrix} 0 & 0 & 0 & \dots & 0   \\ A_\Delta \bar{B} & 0 & 0 & \dots & 0 \\ A_\Delta^2\bar{B} & A_\Delta\bar{B}  & 0 & \dots & 0\\
    \vdots & \vdots & \vdots & \ddots & \vdots \\
    A_\Delta^{\bar{N}-1} \bar{B} & A_\Delta^{\bar{N}-2} \bar{B} & \dots & A_\Delta\bar{B} & 0 \end{bmatrix} \in \mathbb{R}^{d\bar{N} \times m\bar{N}},\\
    \mathbf{A}^ w & =     \begin{bmatrix}
    I_d & 0 & 0 & \dots & 0 \\
    A_\Delta & I_d & 0 & \dots & 0\\
    A_\Delta^2 & A_\Delta & I_d & \dots & 0\\
    \vdots & \vdots & \vdots & \ddots & \vdots \\
    A_\Delta^{\bar{N}-1} &     A_\Delta^{\bar{N}-2} & \dots & \dots & I_d 
    \end{bmatrix} \in \mathbb{R}^{d\bar{N} \times d\bar{N}},
\end{align*}
where $A_\Delta = (\bar{A} + \Delta_A) \in \mathcal{P}_{A_\Delta}$ and $B_\Delta = (\bar{B} + \Delta_B) \in \mathcal{P}_{B_\Delta}$. We write matrices $\bar{\mathbf{A}}_1$ and $\mathbf{A}_\delta \in \mathbb{R}^{d\bar{N} \times d\bar{N}}$ as:
% Also the matrices $\{A_1, A_2, \dots, A_{N-1}\}$ defining $\bar{\mathbf{A}}_v$ are given as 
\begin{align*}
   & \bar{\mathbf{A}}_1 = \begin{bmatrix} I_d & 0 & 0 & \dots & 0\\ \bar{A} & I_d & 0 & \dots & 0 \\ \bar{A}^2 & \bar{A} & I_d & \dots & 0 \\ \vdots & \vdots & \vdots & \ddots & \vdots \\ \bar{A}^{\bar{N}-1} & \bar{A}^{\bar{N}-2} & \dots & \dots & I_d \end{bmatrix},~ \mathbf{A}_\delta = (\mathbf{A}^w - \bar{\mathbf{A}}_1),
\end{align*}
which gives $\mathbf{A}^x =  \bar{\mathbf{A}} + \Big (\bar{\mathbf{A}}_1 + \mathbf{A}_\delta \Big ) \mathbf{\Delta}_A, \mathbf{A}^u =  \bar{\mathbf{B}} + \Big ( \bar{\mathbf{A}}_1 + \mathbf{A}_\delta \Big ) \mathbf{\Delta}_B$, and $\mathbf{A}^{\Delta u} = \Big ( \bar{\mathbf{A}}_1 - \mathbf{I}_d + \mathbf{A}_\delta  \Big) \bar{\mathbf{B}}$.
%%%%%%%%%%%%%%%%%%%%%%%%%%%%%%%%%%%%%%%%%%
The matrix $\bar{\mathbf{A}}_v$ is written as $\bar{\mathbf{A}}_v = \begin{bmatrix} A^{(1)}_v & A^{(2)}_v & \dots & A^{(\bar{N}-1)}_v \end{bmatrix}$, where matrices $\{A^{(1)}_v, A^{(2)}_v, \dots, A^{(\bar{N}-1)}_v\}$ are given as 
\begin{equation*}
\begin{aligned}
   & A^{(1)}_v = \begin{bmatrix} 0 & 0 & 0 & \dots & 0\\ I_d & 0 & 0 & \dots & 0\\ 0 & I_d & 0 & \dots & 0 \\ \vdots & \vdots & \vdots & \ddots & \vdots \\ 0 & 0 & \dots & I_d & 0 \end{bmatrix},~A^{(2)}_v = \begin{bmatrix} 0 & 0 & 0 & \dots & 0 \\ 0 & 0 & 0 & \dots & 0 \\ I_d & 0 & 0 & \dots & 0 \\ 0 & I_d & 0 & \dots & 0\\ \vdots & \vdots & \vdots & \ddots & \vdots \\ 0 & 0 & I_d & \dots & 0  \end{bmatrix},~\textnormal{and analogously for}~A^{(3)}_v, \dots, A^{(\bar{N}-1)}_v.
\end{aligned}
\end{equation*}
This gives $\mathbf{A}^w  = \mathbf{I}_d + \bar{\mathbf{A}}_v \mathbf{A}_\Delta$, with $\mathbf{I}_d = (I_{\bar{N}} \otimes I_d)$, and 
\begin{align}\label{stacked_matrix}
    \mathbf{A}_\Delta = \begin{bmatrix} I_{\bar{N}} \otimes A_\Delta \\ I_{\bar{N}} \otimes A_\Delta^2 \\ \vdots \\ I_{\bar{N}} \otimes A_\Delta^{\bar{N}-1} \end{bmatrix} \in \mathbb{R}^{d\bar{N}(\bar N-1) \times d \bar N}.
\end{align}
%%%
% For any $N_t \in \{1,2,\dots, N\}$ the matrix $G^{(N_t)}_\mathrm{eq} \in \mathbb{R}^{(N_t + 1)d \times ((N_t + 1)d + N_tm)}$ in \eqref{eq:MPC_R_fin_trac} is given by:
% \begin{align*}
%     G^{(N_t)}_\mathrm{eq} = \begin{bmatrix} I_d & 0 & 0 & \cdots & 0 & 0 & 0 & 0 & \cdots & 0 \\ -\bar{A} & I_d & 0 & \cdots & 0 & 0 & -\bar{B} & 0 & \cdots & 0 \\ 0 & -\bar{A} & I_d & \cdots & 0 & 0 & 0 & -\bar{B} & \cdots & 0 \\ \vdots & \vdots & \vdots & \ddots & \vdots & \vdots & \vdots & \vdots & \ddots & \vdots \\ 0 & 0 & 0 & \cdots & -\bar{A} & I_d & 0 & 0 & \cdots & -\bar{B} \end{bmatrix}.  
% \end{align*}
%%
\subsection{Deriving \eqref{eq:fin_state_con} from \eqref{ugo_wants_it_state}}\label{GIVEMEP}
Using \eqref{simplified_dyn_matrices} in \eqref{eq:state_propagation}, constraints \eqref{ugo_wants_it_state} can be written as:
\begin{align}\label{maincon_in_bounding}
& \mathbf{F}^x \Bigg (\bar{\mathbf{A}} \bar{\mathbf{x}}_t + \bar{\mathbf{A}}_1 \mathbf{\Delta}_A \bar{\mathbf{x}}_t  + (\mathbf{A}_\delta \mathbf{\Delta}_A) \bar{\mathbf{x}}_t + \bar{\mathbf{B}} \mathbf{u}_t + \bar{\mathbf{A}}_1 \mathbf{\Delta}_B \mathbf{u}_t +  (\mathbf{A}_\delta \mathbf{\Delta}_B ) \mathbf{u}_t + (\bar{\mathbf{A}}_1 - \mathbf{I}_d + \mathbf{A}_\delta )\bar{\mathbf{B}} \Delta \mathbf{u}_t + \cdots \nonumber \\ 
&~~~~~~~~~~~~~~~~~~~~~~~~~~~~~~~~~~~~~~~~~~~~~~~~~~~~~~~~~~~~~~~~~~~~~~~~~~~~~~~~~~~~~~ + \mathbf{w}_t + \bar{\mathbf{A}}_v \mathbf{A}_\Delta \mathbf{w}_t \Bigg ) \leq \mathbf{f}^x,\\
& \forall \Delta_A \in \mathcal{P}_A,~ \forall \Delta_B \in \mathcal{P}_B,~\forall w_t \in \mathbb{W}. \nonumber 
\end{align}
% where $F^x = \mathrm{diag}(I_{N-1} \otimes H^x, H^x_N) \in \mathbb{R}^{(r(N-1)+r) \times dN}$ and $f^x = [(h^x)^\top, (h^x)^\top, \dots, (h_N^x)^\top ]^\top \in \mathbb{R}^{r(N-1)+r}$. 
% Inequality \eqref{maincon_in_bounding} must hold for all rows $i \in \{1,2,\dots,r(N-1) + r_N\}$. 
We obtain an upper bound for the left hand side of inequality \eqref{maincon_in_bounding} row-wise as follows:
\begin{align}\label{eq:ineq_main_state}
& \mathbf{F}^x_i (\bar{\mathbf{A}} \bar{\mathbf{x}}_t + \bar{\mathbf{B}} \mathbf{u}_t + (\bar{\mathbf{A}}_1 - \mathbf{I}_d) \bar{\mathbf{B}} \Delta \mathbf{u}_t + \mathbf{w}_t) + \mathbf{F}^x_i \bar{\mathbf{A}}_1 \mathbf{\Delta}_A \bar{\mathbf{x}}_t + \mathbf{F}^x_i \bar{\mathbf{A}}_1 \mathbf{\Delta}_B \mathbf{u}_t + \mathbf{F}^x_i \mathbf{A}_\delta \mathbf{\Delta}_A \bar{\mathbf{x}}_t + \cdots \nonumber \\
& ~~~~~~~~~~~~~~~~~~~~~~~~~~~~~~~~~~~~~~~~~~~~~~~~~~~ + \mathbf{F}^x_i \mathbf{A}_\delta \mathbf{\Delta}_B \mathbf{u}_t + \mathbf{F}^x_i \mathbf{A}_\delta \bar{\mathbf{B}}  \Delta \mathbf{u}_t + \mathbf{F}^x_i \bar{\mathbf{A}}_v \mathbf{A}_\Delta \mathbf{w}_t, \nonumber \\[2ex]
& \leq \mathbf{F}^x_i (\bar{\mathbf{A}} \bar{\mathbf{x}}_t + \bar{\mathbf{B}} \mathbf{u}_t + (\bar{\mathbf{A}}_1 - \mathbf{I}_d) \bar{\mathbf{B}} \Delta \mathbf{u}_t + \mathbf{w}_t) + \mathbf{F}^x_i \bar{\mathbf{A}}_1 \mathbf{\Delta}_A  \bar{\mathbf{x}}_t + \mathbf{F}^x_i \bar{\mathbf{A}}_1 \mathbf{\Delta}_B \mathbf{u}_t +\Vert \mathbf{F}^x_i \mathbf{A}_\delta \mathbf{\Delta}_A \Vert_* \Vert \bar{\mathbf{x}}_t \Vert + \nonumber \\
&~~~~~~~~~~~~~~~~~~~~~~~~~~~~~~~~~~~~ \cdots + \Vert \mathbf{F}^x_i \mathbf{A}_\delta \mathbf{\Delta}_B \Vert_* \Vert \mathbf{u}_t \Vert + \Vert \mathbf{F}^x_i \mathbf{A}_\delta \bar{\mathbf{B}} \Vert_* \Vert \Delta \mathbf{u}_t \Vert + \Vert \mathbf{F}^x_i \bar{\mathbf{A}}_v \mathbf{A}_\Delta \Vert_* \Vert \mathbf{w}_t \Vert,
\end{align}
for rows $i \in \{1,2,\dots,r(\bar{N}-1) + r_N\}$, where we have used the H{\"o}lder's inequality. Using bounds \eqref{eq:bound_mainterm}-\eqref{fourthbound} in \eqref{eq:ineq_main_state} then yields \eqref{eq:fin_state_con}. Note that in \eqref{eq:fin_state_con} $\bar{N} \leftarrow N$.
\subsection{Bounding Nominal Trajectory Perturbations}\label{sec:bounds}
For any horizon length\footnote{Note, also the bounds in Section~\ref{sec:opt_tight} were introduced with a fixed horizon length of $N$, i.e., $\bar{N} \leftarrow N$.} of $\bar{N} \in \{2,3,\dots, N\}$, we first bound:
% \begin{subequations}\label{eq:delAdelBbounds}
% \begin{align}
%     \max_{\Delta_A \in \Phi_{A,2}} \Vert F^x_i \mathbf{\Delta}_A \Vert_1 = \mathbf{t}^i_{\delta A},\\
%     \max_{\Delta_B \in \Phi_{B,2}} \Vert F^x_i \mathbf{\Delta}_B \Vert_1 = \mathbf{t}^i_{\delta B},
% \end{align}
% \end{subequations}
% and also the term
\begin{align}\label{adlta}
    & \max_{A_\Delta \in \mathcal{P}_{A_\Delta}} \Vert \mathbf{F}^x_i \mathbf{A}_\delta \Vert_*,~\textnormal{where using \eqref{padelta} we have} \nonumber \\ 
    & \mathbf{A}_\delta = \bar{\mathbf{A}}_v \begin{bmatrix} I_{\bar{N}} \otimes (A_\Delta-\bar{A}) \\ I_{\bar{N}} \otimes (A_\Delta^2 - \bar{A}^2) \\ \vdots \\ I_{\bar{N}} \otimes (A_\Delta^{\bar{N}-1} - \bar{A}^{\bar{N}-1}) \end{bmatrix}.
\end{align}
Note that for all $A_\Delta \in \mathcal{P}_{A_\Delta} \implies A_\Delta^n \in \mathcal{P}^n_{A_\Delta}$, for $n \in \{1,2,\dots, \bar{N}-1\}$, 
% However, since these bounds are attained at the boundaries of the convex sets $\{\Phi^s_{A,2}, (\Phi^s_{A,2})^2, \dots, (\Phi^s_{A,2})^{N-1}\}$ and $\Phi_{B,2}$, one alternative way of getting sufficient upper bounds would be to use the fact that 
% \begin{align}\label{outer_pol_set_disc}
%     \Phi^n_{A,2} \subseteq \mathcal{P}^n_A,~\textnormal{for } n=1,2,\dots, N-1,
% \end{align}
% and then to evaluate the values of the terms bounded at the extreme points, i.e., vertices of the sets $\mathcal{P}^n_A$. 
where $\mathcal{P}^n_{A_\Delta}$ is the set of all matrices that can be written as a convex combination of matrices obtained with the product of \emph{all possible combinations} of $n$ matrices out of $\{(\bar{A} + \Delta^{(1)}_A), (\bar{A}+\Delta^{(2)}_A), \dots, (\bar{A} + \Delta^{(n_a)}_A) \}$. 
% , with repetitions allowed. 
Hence
\begin{align}\label{eq:bound_mainterm}
    & \max_{A_\Delta \in \mathcal{P}_{A_\Delta}} \Vert \mathbf{F}^x_i \mathbf{A}_\delta \Vert_* \leq 
    \max_{\substack{\Delta_1 \in \mathcal{P}_{A_\Delta} \\ \Delta_2 \in \mathcal{P}^2_{A_\Delta} \\ \vdots \\ \Delta_{\bar{N}-1} \in \mathcal{P}^{\bar{N}-1}_{A_\Delta} }} \Vert \mathbf{F}^x_i \bar{\mathbf{A}}_v \begin{bmatrix} I_N \otimes (\Delta_1-\bar{A}) \\ I_N \otimes (\Delta_2-\bar{A}^2) \\ \vdots \\ I_N \otimes (\Delta_{\bar{N}-1}-\bar{A}^{\bar{N}-1}) \end{bmatrix}\Vert_* = \mathbf{t}^i_0,
\end{align}
where we have relaxed all the equality constraints among the matrices $\{\Delta_1, \Delta_2,\dots, \Delta_{\bar{N}-1}\}$. Using the above bound \eqref{eq:bound_mainterm}, we get 
\begin{align}\label{eq:bound_firstterm}
    & \max_{\substack{A_\Delta \in \mathcal{P}_{A_\Delta} \\\Delta_A \in \mathcal{P}_A}} \Vert \mathbf{F}^x_i \mathbf{A}_\delta  \mathbf{\Delta}_A \Vert_* 
    % = \max_{\substack{\Delta_1 \in \Phi_A \\ \Delta_2 \in \Phi_A^2 \\ \Delta_3 \in \Phi_A^3 \\ \vdots \\ \Delta_1^2 = \Delta_2 \\ \Delta_1 \Delta_2 = \Delta_3 \\ \vdots  }}\Vert\mathbf{F}^x_i \bar{\mathbf{A}}_v \begin{bmatrix} I_{N} \otimes \Delta_1 \\ I_N \otimes \Delta_2 \\ I_N \otimes \Delta_3 \\ \vdots \\ I_N \otimes \Delta_N \end{bmatrix}\Vert_1 
    % \leq  
    % \max_{\substack{\Delta_1 \in \Phi_A \\ \Delta_2 \in \Phi_A^2 \\ \Delta_3 \in \Phi_A^3 \\ \vdots \\ \Delta_1^2 \leq \Delta_2 \\ \Delta_1 \Delta_2 \leq \Delta_3 \\ \vdots }}\VertF_i \mathbf{A} \begin{bmatrix} I_{N} \otimes \Delta_1 \\ I_N \otimes \Delta_2 \\ I_N \otimes \Delta_3 \\ \vdots \end{bmatrix}\Vert_1 
    \leq 
    \mathbf{t}^i_0  \max_{\Delta_A \in \mathcal{P}_A} \Vert \mathbf{\Delta}_A\Vert_p = \mathbf{t}^i_1,
\end{align}
where we have used the consistency property of induced norms, for any $p = 1,2,\infty$.
%%%%%%%%%%%%%%%%%%%%%%%%%%%%%%%%%%%%%
Similarly, bounding terms 
\begin{align}\label{eq:secondbound}
    \max_{\substack{A_\Delta \in \mathcal{P}_{A_\Delta}\\ \Delta_B \in \mathcal{P}_B}}\Vert \mathbf{F}^x_i \mathbf{A}_\delta \mathbf{\Delta}_B \Vert_* 
    \leq   \mathbf{t}^i_0  \max_{\Delta_B \in \mathcal{P}_B} \Vert \mathbf{\Delta}_B \Vert_p = \mathbf{t}^i_2,
\end{align}
%%%%%%%%%%%%%%%%%%%%%%%%%%%%%
and
\begin{align}\label{eq:boundthird}
    \max_{A_\Delta \in \mathcal{P}_{A_\Delta}} \Vert \mathbf{F}^x_i \mathbf{A}_\delta \bar{\mathbf{B}}  \Vert_* 
    \leq \mathbf{t}^i_0 \Vert \bar{\mathbf{B}} \Vert_p = \mathbf{t}^i_3,
\end{align}
and finally
\begin{align}\label{fourthbound}
& \max_{A_\Delta \in \mathcal{P}_{A_\Delta}} \Vert \mathbf{F}^x_i \bar{\mathbf{A}}_v \mathbf{A}_\Delta \Vert_* \leq \max_{\substack{\Delta_1 \in \mathcal{P}_{A_\Delta} \\ \Delta_2 \in \mathcal{P}^2_{A_\Delta} \\ \vdots \\ \Delta_{\bar{N}-1} \in \mathcal{P}^{\bar{N}-1}_{A_\Delta}}}\Vert \mathbf{F}^x_i \bar{\mathbf{A}}_v \begin{bmatrix} I_{\bar{N}} \otimes \Delta_1 \\ I_{\bar{N}} \otimes \Delta_2 \\ \vdots \\ I_{\bar{N}} \otimes \Delta_{\bar{N}-1} \end{bmatrix} \Vert_* = \mathbf{t}^i_w, 
\end{align}
% Problems \eqref{eq:bound_mainterm}-\eqref{fourthbound} are to be solved only once before the control process is started at time step $t=0$. 
for $i \in \{1,2,\dots, r(\bar{N}-1)+r_N\}$. Problems \eqref{eq:bound_mainterm}-\eqref{fourthbound} are maximizing convex functions of the decision variables over convex and compact domains. Therefore, these maximum bounds are attained at the extreme points, i.e., vertices of the convex sets $\{\mathcal{P}_{A_\Delta}, \mathcal{P}^2_{A_\Delta}, \dots, \mathcal{P}^{\bar{N}-1}_{A_\Delta}\}$, $\mathcal{P}_A$ and $\mathcal{P}_B$. Consequently, the optimal values of \eqref{eq:bound_mainterm}-\eqref{fourthbound} can be obtained by evaluating the values of each of the terms in \eqref{eq:bound_mainterm}-\eqref{fourthbound} at all possible combinations of such extreme points. Since such a vertex enumeration strategy scales poorly with the horizon length $N$, a computationally cheaper alternative to bounds \eqref{eq:bound_mainterm}-\eqref{fourthbound} is presented next. 
% in \cite{bujarbaruah2020robust}. 
%%
\subsection{Computationally Efficient Alternatives of Bounds \eqref{eq:bound_mainterm}-\eqref{fourthbound}}\label{bin_bound}
Recall the optimization problem from \eqref{eq:bound_mainterm}, given by \vspace{3pt}
\begin{align}\label{eq:term_bound}
    & \max_{A_\Delta \in \mathcal{P}_{A_\Delta}} \Vert \mathbf{F}^x_i \mathbf{A}_\delta \Vert_*,~\textnormal{with $\mathbf{A}_\delta$ from \eqref{adlta}}.
\end{align}
Using the triangle and H{\"o}lder's inequalities, and the submultiplicativity and consistency properties of induced norms, \eqref{eq:term_bound} can be upper bounded for any cut-off horizon $\tilde{N}<\bar{N}$ as follows:
\begin{align}\label{eq:bound_maintermD}
    & \max_{A_\Delta \in \mathcal{P}_{A_\Delta}} \Vert \mathbf{F}^x_i \mathbf{A}_\delta \Vert_* \leq \tilde{\mathbf{t}}^i_0 + \hat{\mathbf{t}}^i_0 = \mathbf{t}^i_0,
\end{align}
with
\begin{align*}
\tilde{\mathbf{t}}^i_0 = \max_{\substack{\Delta_1 \in \mathcal{P}_{A_\Delta} \\ \vdots \\ \Delta_{\tilde{N}-1} \in \mathcal{P}^{\tilde{N}-1}_{A_\Delta} }} \Vert \mathbf{F}^x_i \bar{\mathbf{A}}^{1:(\tilde{N}-1)}_v \begin{bmatrix} I_{\bar{N}} \otimes (\Delta_1-\bar{A}) \\ I_{\bar{N}} \otimes (\Delta_2-\bar{A}^2) \\ \vdots \\ I_{\bar{N}} \otimes (\Delta_{\tilde{N}-1}-\bar{A}^{\tilde{N}-1}) \end{bmatrix}\Vert_*
\end{align*}
where $\bar{\mathbf{A}}^{n_1:n_2}_v$ denotes $\begin{bmatrix} A^{(n_1)}_v & A^{(n_1 + 1)}_v & \dots & A^{(n_2)}_v \end{bmatrix}$, with the associated matrices defined in Appendix~\ref{A1}, $\mathbf{F}_i^x[n_1:n_2]$ denotes the $n_1$ to $n_2$ columns of the row vector $\mathbf{F}^x_i$, for $i \in \{1,2,\dots, r(\bar{N}-1)+r_N\}$, and  
\begin{align*}
    &\hat{\mathbf{t}}^i_0 = \max_{\Delta_A \in \mathcal{P}_A} \Bigg ( \sum_{j=\tilde{N}+1}^{\bar{N}} \Vert \mathbf{F}^x_i [(j-1)d +1 : jd] \Vert_* \big ( \sum_{k=1}^{j-\tilde{N}} ( \sum_{l = 1}^{j-k} {j-k \choose l} \Vert \bar{A} \Vert^{j-k-l}_p \Vert \Delta_A \Vert^l_p ) \big )  \Bigg ).
\end{align*}
Using the above derived bound \eqref{eq:bound_maintermD} we obtain: 
\begin{align*}
  & \max_{\substack{A_\Delta \in \mathcal{P}_{A_\Delta} \\\Delta_A \in \mathcal{P}_A}} \Vert \mathbf{F}^x_i \mathbf{A}_\delta \mathbf{\Delta}_A \Vert_*
    \leq \mathbf{t}^i_0  \max_{\Delta_A \in \mathcal{P}_{A}} \Vert \mathbf{\Delta}_A \Vert_p
  = \mathbf{t}^i_1,
\end{align*}
where we have used the consistency property of induced norms, for any $p = 1,2,\infty$.
%%%%%%%%%%%%%%%%%%%%%%%%%%%%%%%%%%%%%
Similarly, we bound 
\begin{align*}
    \max_{\substack{A_\Delta \in \mathcal{P}_{A_\Delta}\\ \Delta_B \in \mathcal{P}_B}}\Vert \mathbf{F}^x_i \mathbf{A}_\delta \mathbf{\Delta}_B \Vert_* 
    \leq   \mathbf{t}^i_0  \max_{\Delta_B \in \mathcal{P}_B} \Vert \mathbf{\Delta}_B \Vert_p = \mathbf{t}^i_2,
\end{align*}
and, 
\begin{align*} %\label{eq:boundthirdD}
    & \max_{A_\Delta \in \mathcal{P}_{A_\Delta}} \Vert \mathbf{F}^x_i {\mathbf{A}}_\delta \bar{\mathbf{B}} \Vert_*
    \leq \mathbf{t}^i_0 \Vert \bar{\mathbf{B}} \Vert_p 
    = \mathbf{t}^i_3,
\end{align*}
and finally using $\mathbf{A}_\Delta$ from \eqref{stacked_matrix}
\begin{align*} %\label{fourthboundD}
    & \max_{A_\Delta \in \mathcal{P}_{A_\Delta}} \Vert \mathbf{F}^x_i \bar{\mathbf{A}}_v \mathbf{A}_\Delta \Vert_* \leq \tilde{\mathbf{t}}^i_w + \hat{\mathbf{t}}^i_w = {\mathbf{t}}^i_w,  
\end{align*}
for all $i \in \{1,2,\dots, r(\bar{N}-1)+r_N\}$, where 
\begin{align*}
    & \tilde{\mathbf{t}}^i_w =  \max_{\substack{\Delta_1 \in \mathcal{P}_{A_\Delta} \\ \vdots \\ \Delta_{\tilde{N}-1} \in \mathcal{P}^{\tilde{N}-1}_{A_\Delta} }} \Vert \mathbf{F}^x_i \bar{\mathbf{A}}^{1:(\tilde{N}-1)}_v \begin{bmatrix} I_{\bar{N}} \otimes \Delta_1 \\ I_{\bar{N}} \otimes \Delta_2 \\ \vdots \\ I_{\bar{N}} \otimes \Delta_{\tilde{N}-1} \end{bmatrix}\Vert_*,
    \end{align*}
and 
\begin{align*}
 & \hat{\mathbf{t}}^i_w = \max_{\Delta_A \in \mathcal{P}_{A}} \Bigg ( \sum_{j=\tilde{N}}^{\bar{N}-1} \Vert \mathbf{F}^x_i A^{(j)}_v \Vert_* \Big ( \Vert (I_{\bar{N}} \otimes \bar{A})^j \Vert_p + \sum_{k=1}^{j} {j \choose k} \Vert (I_{\bar{N}} \otimes \bar{A})  \Vert^{j-k}_p \Vert (I_{\bar{N}} \otimes \Delta_A)\Vert^{k}_p \Big ) \Bigg ), 
\end{align*}
where we have used the property of two matrices $X$ and $Y$ yielding:
\begin{align*}
 &\Vert (X+ Y)^n \Vert_p \leq  \Vert X^n \Vert_p + \sum_{k = 1}^{n} {n \choose k} \Vert X \Vert_p^{n-k} \Vert Y \Vert_p^{k},\\
 &\forall n \in \{\tilde{N},\tilde{N}+1, \dots, \bar{N}-1 \}.
\end{align*}
This cut-off horizon $\tilde{N}$ can be chosen based on the available computational resources at the expense of more  conservatism over \eqref{eq:bound_mainterm}-\eqref{fourthbound}.
\subsection{Obtaining \eqref{eq:state_robcon} from \eqref{eq:fin_state_con}}\label{Ap_bnd_der}
Here we derive \eqref{eq:state_robcon} from \eqref{eq:fin_state_con}. Using bounds \eqref{tdel1_etc} and \eqref{eq:bound_mainterm}-\eqref{fourthbound} and policy parametrization \eqref{eq:inputParam_DF_OL}, constraints \eqref{eq:fin_state_con} can be satisfied by imposing:
\begin{align}\label{eq:rob_statecon}
    &  \max_{\mathbf{w}_t \in \mathbf{W}} \Bigg ( \mathbf{F}^x_i(\bar{\mathbf{A}} \bar{\mathbf{x}}^{(N_t)}_t + \bar{\mathbf{B}} (\mathbf{M}^{(N_t)}_t \mathbf{w}_t + \bar{\mathbf{u}}^{(N_t)}_t) + (\bar{\mathbf{A}}_1 - \mathbf{I}_d)\bar{\mathbf{B}}\mathbf{M}^{(N_t)}_t \mathbf{w}_t + \mathbf{w}_t) + \mathbf{t}_{\delta 1}^{(N_t), i} \Vert \bar{\mathbf{x}}_t^{(N_t)} \Vert + \cdots \nonumber \\
    &~~~~~~~~~~~~~ + (\mathbf{t}_{2}^{(N_t), i} + \mathbf{t}_{\delta B}^{(N_t), i} ) \Vert \mathbf{M}^{(N_t)}_t \mathbf{w}_t + \bar{\mathbf{u}}^{(N_t)}_t \Vert + \mathbf{t}_{3}^{(N_t), i} \Vert \mathbf{M}^{(N_t)}_t \mathbf{w}_t \Vert  + \mathbf{t}_{w}^{(N_t), i} \mathbf{w}_\mathrm{max} \Bigg ) \leq \mathbf{f}^x_i,
\end{align}
where using \eqref{tdel1_etc} we have used the H{\"o}lder's and the triangle inequality to bound $\mathbf{F}^x_i \bar{\mathbf{A}}_1 \mathbf{\Delta}_A \bar{\mathbf{x}}^{(N_t)}_t$ and $\mathbf{F}^x_i \bar{\mathbf{A}}_1\mathbf{\Delta}_B (\mathbf{M}^{(N_t)}_t \mathbf{w}_t + \bar{\mathbf{u}}^{(N_t)}_t)$ for all rows $i \in \{1,2,\dots, r(N_t-1)+r_N\}$. 
Use the induced norm consistency property and the triangle inequality in \eqref{eq:rob_statecon} as:
\begin{equation}\label{xxxx}
\begin{aligned}
    & (\mathbf{t}_{2}^{(N_t), i} + \mathbf{t}_{\delta B}^{(N_t), i} ) \Vert \mathbf{M}^{(N_t)}_t \mathbf{w}_t + \bar{\mathbf{u}}^{(N_t)}_t \Vert +  \mathbf{t}^{(N_t)}_3 \Vert \mathbf{M}^{(N_t)}_t \mathbf{w}_t  \Vert,  \\
    & \leq (\mathbf{t}_{2}^{(N_t), i} + \mathbf{t}_{\delta B}^{(N_t), i} + \mathbf{t}_{3}^{(N_t), i} ) \Vert \mathbf{M}^{(N_t)}_t \Vert_p \mathbf{w}_\mathrm{max} + (\mathbf{t}_{2}^{(N_t), i} + \mathbf{t}_{\delta B}^{(N_t), i} ) \Vert \bar{\mathbf{u}}^{(N_t)}_t \Vert,\\
    & \leq \mathbf{t}_{\delta 3}^{(N_t), i} \Vert \mathbf{M}^{(N_t)}_t \Vert_p \mathbf{w}_\mathrm{max} + \mathbf{t}_{\delta 2}^{(N_t), i} \Vert \bar{\mathbf{u}}^{(N_t)}_t \Vert,
\end{aligned}
\end{equation}
for any $p = 1,2,\infty$, where we have used the definitions \eqref{tdel1_etc}. Using \eqref{xxxx} in \eqref{eq:rob_statecon} for all rows $i \in \{1,2,\dots, r(N_t-1)+r_N\}$, we define
\begin{align*}
% & f^x_\mathrm{tight} \nonumber \\
\!\mathbf{f}^x_\mathrm{tight}& \!=\! \mathbf{f}^x\! -\! \mathbf{t}^{(N_t)}_{\delta 1} \Vert \bar{\mathbf{x}}_t^{(N_t)} \Vert \!-\! \mathbf{t}^{(N_t)}_{\delta 3} \Vert \mathbf{M}^{(N_t)}_t \Vert_p \mathbf{w}_\mathrm{max} - \mathbf{t}^{(N_t)}_{\delta 2} \Vert \bar{\mathbf{u}}^{(N_t)}_t \Vert - \mathbf{t}^{(N_t)}_w \mathbf{w}_\mathrm{max},
\end{align*}
which yields \eqref{eq:state_robcon} with tightened constraints \eqref{tight1}.
\subsection{Reformulation of \eqref{eq:MPC_R_fin_trac} via Duality of Convex Programs}\label{duality}
We again consider the following two cases for satisfying the robust state constraints \eqref{eq:state_robcon}.
\vspace{3pt}\\
\noindent \textbf{Case 1}: ($N_t \geq 2$, i.e., $ t \leq N-2$) 
Constraints \eqref{Ng1_con} can be satisfied using duality of convex programs by solving:
\begin{equation*}
\begin{aligned}
&     \mathbf{F}^x (\bar{\mathbf{A}} \bar{\mathbf{x}}^{(N_t)}_t + \bar{\mathbf{B}}\bar{\mathbf{u}}^{(N_t)}_t) + \Lambda^{(N_t)} \mathbf{h}^w \leq \mathbf{f}_\mathrm{tight}^x, \\[1ex]
& \Lambda^{(N_t)} \geq 0, \\
% & \Vert p^{(N_t)}_i \Vert_* \leq \mathbf{t}^{(N_t), i}_0,~\forall i \in \{1,2,\dots, r(N_t-1)+r_N\}, \nonumber \\
&   \Lambda^{(N_t)} \mathbf{H}^w = \Big (\mathbf{F}^x(\bar{\mathbf{B}}\mathbf{M}^{(N_t)}_t + (\bar{\mathbf{A}}_1 - \mathbf{I}_d) \bar{\mathbf{B}}\mathbf{M}^{(N_t)}_t + \mathbf{I}_d) \Big ), 
\end{aligned}
\end{equation*}
where $\mathbf{f}^x_\mathrm{tight}$ is obtained from \eqref{tight1}, and dual variables $\Lambda^{(N_t)}  \in \mathbb{R}^{(r(N_t-1)+r_N) \times aN_t}$.
\vspace{3pt}\\
\noindent \textbf{Case 2}: ($N_t = 1$, i.e., $ t \geq N-1$)
Consider the case of $N_t=1$. As pointed out in \eqref{N1_con}, the robust state constraint for this case can be simplified and written as
\begin{align*}
    \max_{\substack{{w}_t \in \mathbb{W}\\ \Delta_A \in \mathcal{P}_A \\ \Delta_B \in \mathcal{P}_B}}  H^x_N((\bar{{A}}+{\Delta}_A) \bar{\mathbf{x}}^{(1)}_t + (\bar{{B}} + {\Delta}_B) \bar{\mathbf{u}}^{(1)}_t + {w}_t) \leq h^x_N, 
\end{align*}
which we must solve \emph{exactly} (i.e., find $h^x_N$ where the max is attained) for the uncertainty representation ${w}_t \in \mathbb{W},~ \Delta_A \in \mathcal{P}_A$ and $\Delta_B \in \mathcal{P}_B$, in order for guarantees of Theorem~\ref{thm1} to hold. Using duality of convex programs \cite{ben2009robust} one can write the robust state constraints \eqref{N1_con} equivalently as: 
\begin{align}\label{n1dual}
&  H^x_N ((\bar{{A}} + \Delta^{(j)}_A) \bar{\mathbf{x}}^{(1)}_t +  (\bar{{B}} + \Delta^{(k)}_B) {\bar{\mathbf{u}}}^{(1)}_t) + \Lambda^{(1)} h^w\leq h^x_N, \nonumber \\
& \Lambda^{(1)} \geq 0,~H^x_N = \Lambda^{(1)}H^w, \nonumber \\
& \forall j \in \{1,2,\dots, n_a\},~\forall k \in \{1,2,\dots, n_b\},
\end{align}
where dual variables $\Lambda^{(1)} \in \mathbb{R}^{r_N \times a}$. \vspace{3pt} \\
\noindent \textbf{Input Constraints}: 
Considering the robust input constraints \eqref{eq:input_robcon} for any $N_t \in \{1,2,\dots,N\}$, one can similarly show that this is equivalent to:
\begin{equation*}
\begin{aligned}
  & (\gamma^{(N_t)})^\top \mathbf{h}^w \leq \mathbf{h}^u - \mathbf{H}^u \bar{\mathbf{u}}^{(N_t)}_t,\\
  & (\mathbf{H}^u\mathbf{M}^{(N_t)}_t)^\top
 = (\mathbf{H}^w)^\top \gamma^{(N_t)},~\gamma^{(N_t)} \geq 0,
 \end{aligned}
\end{equation*}
by introducing decision variables of $\gamma^{(N_t)} \in \mathbb{R}^{aN_t \times oN_t}$ in \eqref{eq:MPC_R_fin_trac} for each horizon length $N_t \in \{1,2,\dots,N\}$. 
%%%%%%%%%%%%%%%%%%%%%%%%%%%%%%%%%%%%%%%
\subsection{Proof of Theorem~\ref{thm1}}\label{ProofThe1}
By assumption, at time step $t=0$ problem \eqref{eq:MPC_R_fin_trac} with tightened constraints \eqref{tight1} is feasible, with a horizon length $N_t = N$. We then prove robust satisfaction of \eqref{eq:FTOCP_constr} at all time steps $t \geq 0$ with controller \eqref{eq:cl_control} in closed-loop, by considering the following two cases: \vspace{3pt}\\
\noindent \textbf{Case 1}: ($2 \leq N_t <N$, i.e., $0< t \leq N-2$)
As \eqref{eq:MPC_R_fin_trac} is feasible at time step $t=0$ for $N_t=N$ chosen as per \eqref{eq:n_t}, let the corresponding optimal policy sequence be 
\begin{align}\label{feas_tm1}
\{u^\star_{0|0}, u^\star_{1|0}(\cdot), \dots, u^\star_{N-1|0}(\cdot)\}.
\end{align}
For time steps $t \in \{1,2,\dots, N-2\}$ recall that we define the MPC policy in \eqref{eq:cl_control} as:
\begin{align}\label{sarahpol}
    u^\mathrm{MPC}_t(x_t) = \begin{cases} \bar{u}^\star_{t|t},~\textnormal{if \eqref{eq:MPC_R_fin_trac} is feasible,} \\ u^\star_{t|t_\mathrm{f}}(x_t),~\textnormal{otherwise}, \end{cases}
\end{align}
where $t_\mathrm{f} \in \{0,1,\dots, N-1\}$ is the latest time step when \eqref{eq:MPC_R_fin_trac} was feasible. Policy \eqref{sarahpol} satisfies \eqref{eq:FTOCP_constr} robustly for all $t \in \{1,2,\dots, N-2\}$, as it is a solution to the constrained robust optimal control problem \eqref{eq:MPC_R_fin_trac}. Moreover, from \eqref{feas_tm1}, we have that $t_\mathrm{f} = 0$ is a guaranteed certificate, in case \eqref{eq:MPC_R_fin_trac} continues to be infeasible for all $t \in \{1,2,\dots, N-2\}$.  \vspace{3pt} \\
\noindent \textbf{Case 2}: ($N_t = 1$, i.e., $t \geq N-1$)
Consider the time step $t=N-1$, where from \eqref{eq:n_t} the MPC horizon length $N_t = 1$. In this case we consider constraints \eqref{N1_con} given by:
\begin{align}\label{eq:state_conN1}
    \max_{\substack{{w}_t \in \mathbb{W}\\ \Delta_A \in \mathcal{P}_A \\ \Delta_B \in \mathcal{P}_B}}  H^x_N((\bar{{A}}+{\Delta}_A) \bar{\mathbf{x}}^{(1)}_t + ({\bar{B}} + {\Delta}_B) \bar{\mathbf{u}}^{(1)}_t + {w}_t) \leq h^x_N.
\end{align}
From \eqref{sarahpol} we know that at time step $t=N-1$, there exists a $t_\mathrm{f}$ such that control action $u^{\star}_{t|t_\mathrm{f}}(x_t)$ robustly steers the state $x_t$ to $\mathcal{X}_N$ in one time step. 
% Similar to the previous case, $t_\mathrm{f} = 0$ is a known such  time step, in case \eqref{eq:MPC_R_fin_trac} continues to be infeasible for all $t \in \{1,2,\dots, N-2\}$. 
% Note, the synthesis of this feasible policy $u^\star_{N-1|t_\mathrm{f}}(\cdot)$ involved over-approximating the system uncertainty, see bounds \eqref{eq:bound_mainterm}-\eqref{fourthbound} used in tightened constraint \eqref{Ng1_con}.

Now, at $t=N-1$, we solve \eqref{eq:state_conN1} \emph{exactly} (i.e., find $h^x_N$ where the max is attained) by using duality arguments in \eqref{n1dual}, without any uncertainty over-approximation. Therefore, the optimization problem \eqref{eq:MPC_R_fin_trac} with constraint \eqref{eq:state_conN1} is guaranteed to be feasible at $t=N-1$, with $u^\star_{N-1|t_\mathrm{f}}(\cdot)$ as a feasibility certificate. Let us denote the corresponding optimal policy from $t=N-1$ as: 
\begin{align}\label{lem_proof_pol}
u^\mathrm{MPC}_t(x_t) = \bar{u}^\star_{t|t}.
\end{align}
Let policy \eqref{lem_proof_pol} be applied to \eqref{eq:unc_system} in closed-loop, so that the system reaches the terminal set $\mathcal{X}_N$ at time step $t+1$. Consider solving \eqref{eq:state_conN1} at this step with a horizon length of $N_{t+1}=1$. As, constraint \eqref{eq:state_conN1} uses the same representation of the system uncertainty in satisfying \eqref{eq:FTOCP_constr}-\eqref{FTOCP_termC} robustly as done in \eqref{eq:term_set_DF}, we can infer that a candidate policy at time step $(t+1)$ is 
\begin{align}\label{eq:can_tp1}
    u_{t+1|t+1}(x_{t+1}) = Kx_{t+1},
\end{align}
which is a feasible solution to the robust optimization problem \eqref{eq:MPC_R_fin_trac} under constraint \eqref{eq:state_conN1}. Thus, \eqref{eq:MPC_R_fin_trac} is guaranteed to remain feasible at $(t+1)$ with $N_{t+1}=1$. This completes the proof. $\blacksquare$   

% \balance
%%%%%%%%%%%%%%%%%%%%%%%%%%%%%%%%%
\subsection{Proof of Theorem~\ref{isstheorem}}\label{ProofISS}
% We prove this by considering the following two cases: \vspace{3pt}\\
% \noindent \textbf{Case 1}: ($2 \leq N_t <N$, i.e., $0< t \leq N-2$)

% \vspace{3pt}\\
% \noindent \textbf{Case 2}: ($2 \leq N_t <N$, i.e., $0< t \leq N-2$)
First, we have from Case-2 in the proof of Theorem~\ref{thm1} that at time step $t=N-1$, the problem~\eqref{eq:MPC_R_fin_trac} is feasible with horizon $N_t=1$ and therefore $x_t \in \mathcal{X}_N$ for all $t \geq N$. 

Now, consider the case of $t \geq N$, i.e., $N_t = 1$. 
% From Assumption~\ref{assump:stagecost} we know that, $\alpha_1(\Vert x_t \Vert_2) \leq \ell(x,0) \leq V_{t \rightarrow t+1}^{\mathrm{MPC}}(x_t, \mathbf{t}^{(1)}, 1)$ for some $\alpha_1(\cdot) \in \mathcal{K}_\infty$ and for all $x \in \mathcal{X}_N$. Moreover, 
Since \eqref{eq:MPC_R_fin_trac} for $N_t=1$ can be reformulated into a parametric QP, $V_{t \rightarrow t+1}^{\mathrm{MPC}}(x_t, 0, 1)$ is continuous and piecewise quadratic in $\mathcal{X}_N$ with $V_{t \rightarrow t+1}^{\mathrm{MPC}}(0, 0, 1)= 0$ \cite{bemporad2002explicit}. Hence, under Assumptions~\ref{assump:stagecost}-\ref{assump: termcost}, using the standard proof of \cite[Theorem~23]{Goulart2006}, we conclude that the origin of closed-loop system \eqref{eq:cl_loop_system} is ISS according to Definition~\ref{iss_lyap}, and $V^\mathrm{MPC}_{t \rightarrow t+1}(x_t,0,1)$ is ISS Lyapunov function $\forall t \geq N$. $\blacksquare$

\end{document}